\documentclass{llncs}
\usepackage{amssymb}

\usepackage{amsthm}
\usepackage{afterpage}
\usepackage{mathtools}
\usepackage{graphicx}
\usepackage[ruled,vlined,linesnumbered]{algorithm2e}
\usepackage{url}
\usepackage{float}
\usepackage[usenames,dvipsnames,table]{xcolor}
\usepackage{booktabs}
\usepackage{todonotes}
\usepackage{mdframed}
\usepackage{wrapfig}
\usepackage{rotating}
\usepackage{multirow}
\usepackage{multicol}
\usepackage{tabularx}
\usepackage{longtable}
\usepackage{comment}
\usepackage{caption}
\usepackage{pifont}
\usepackage[T1]{fontenc}
\usepackage{adjustbox}
\usepackage{color}
\usepackage{txfonts}
\usepackage{graphicx}
\usepackage{subfigure}
\usepackage{url,hyperref}
\usepackage{xspace}
\usepackage{comment}
\usepackage{booktabs}
\usepackage{tikz}
\usepackage{siunitx}
\usepackage{array}
\usetikzlibrary{shapes.gates.logic.US,trees,positioning,arrows,spy,patterns}
\pgfdeclarepatternformonly{south west lines}{\pgfqpoint{-0pt}{-0pt}}{\pgfqpoint{3pt}{3pt}}{\pgfqpoint{3pt}{3pt}}{
        \pgfsetlinewidth{0.4pt}
        \pgfpathmoveto{\pgfqpoint{0pt}{0pt}}
        \pgfpathlineto{\pgfqpoint{3pt}{3pt}}
        \pgfpathmoveto{\pgfqpoint{2.8pt}{-.2pt}}
        \pgfpathlineto{\pgfqpoint{3.2pt}{.2pt}}
        \pgfpathmoveto{\pgfqpoint{-.2pt}{2.8pt}}
        \pgfpathlineto{\pgfqpoint{.2pt}{3.2pt}}
        \pgfusepath{stroke}}
\usepackage[dvipsnames]{xcolor}

\usepackage{todonotes}
\usepackage{placeins}
\usepackage{afterpage}
\usepackage{lscape}
\usepackage{comment}
\DeclareMathOperator*{\heureta}{\eta}
\DeclarePairedDelimiter\ceil{\lceil}{\rceil}
\DeclarePairedDelimiter\floor{\lfloor}{\rfloor}

\newcolumntype{B}{>{\hsize=1.5\hsize}X}
\newcolumntype{C}{>{\hsize=1.2\hsize}X}

\newcommand\KleanthiNotes[1]{ }

\newcommand{\limes}{\textsc{Limes}\xspace}
\newcommand{\silk}{\textsc{Silk}\xspace}
\newcommand{\radon}{\textsc{Radon}\xspace}
\newcommand{\strabon}{\textsc{Strabon}\xspace}
\newcommand{\eh}{\textsc{eth}\xspace}

\usepackage{listings} 
\lstdefinestyle{sparql}{float=tb, numberblanklines=true, morekeywords={SELECT,DISTINCT,SAMPLE,FROM,WHERE,FILTER,ORDER,GROUP,BY,IN,AS,GRAPH,SERVICE,PREFIX}}

\tikzset{filled/.style={fill=circle area, draw=circle edge, thick},
    outline/.style={draw=circle edge, thick}}

\begin{document}

\title{Annex: \radon~-- Rapid Discovery of Topological Relations}

%
\author{Mohamed Ahmed Sherif$^{a}$, Kevin Dre\ss{}ler$^{a}$, \\Panayiotis Smeros$^{b}$ and Axel-Cyrille Ngonga Ngomo$^{a}$}
\institute{$^{a}$ Department of Computer Science, University of Leipzig, 04109 Leipzig, Germany \\
\texttt{\{sherif|dressler|ngonga\}@informatik.uni-leipzig.de}\\
$^{b}$ EPFL, BC 142, Station 14, CH-1015 Lausanne, Switzerland\\
\texttt{panayiotis.smeros@epfl.ch}	
}
\maketitle

\begin{abstract}
    
Datasets containing geo-spatial resources are increasingly being represented according to the Linked Data principles. Several time-efficient approach-es for discovering links between RDF resources have been developed over the last years. However, the time-efficient discovery of topological relations between geo-spatial resources has been paid little attention to. We address this research gap by presenting \radon, a novel approach for the rapid computation of topological relations between geo-spatial resources. Our approach uses a sparse tiling index in combination with minimum bounding boxes to reduce the computation time of topological relations. Our evaluation of \radon's runtime on $45$ datasets and in more than $800$ experiments shows that it outperforms the state of the art by up to 3 orders of magnitude while maintaining an F-measure of 100\%. Moreover, our experiments suggest that \radon scales up well when implemented in parallel.
\end{abstract}


\section{Introduction}
\label{sec:introduction}
Geo-spatial datasets belong to the largest sources of Linked Data. For example, \emph{LinkedGeoData} contains more than 20 billion triples which describe millions of geo-spatial entities. 
Datasets such as \emph{NUTS} use polygons of up to $1500$ points to describe resources such as countries. As pointed out in previous works~\cite{orchid}, only 7.1\% of the links between resources connect geo-spatial entities. This is due to two main factors. First, the \emph{large number of geo-spatial resources} available on the Linked Data Web requires scalable algorithms for computing links between geo-spatial resources.  In addition, the \emph{description of geo-spatial resources being commonly based on polygons} demands the computation of particular relations, i.e., topological relations, between geo-spatial resources. 
According to the Linked Data principles\footnote{\url{https://www.w3.org/DesignIssues/LinkedData.html}} and for the sake of real-time application such as \emph{structured machine learning} (e.g., DL-Learner~\cite{dllearner_jmlr}) and \emph{question Answering} (e.g., DEQA platform~\cite{LEH+12b}), the provision of explicit topological relations between resources is of central importance to achieve scalability. 
However, only a few approaches have been developed to deal with geo-spatial data represented in RDF. For example, \cite{orchid} uses the \emph{Hausdorff} distance to compute a topological distance between geo-spatial entities. \cite{Panayiotisldow2016} builds upon \emph{MultiBlock} to compute topological relations according to the DE-9IM standard between geo-spatial entities. 

We go beyond the state of the art by providing \emph{a novel indexing method combined with space tiling that allows for the efficient computation of topological relations between geo-spatial resources}. 
In particular, we present a novel sparse index for geo-spatial resources. We then develop a strategy to discard unnecessary computations for DE-9IM relations based on \emph{bounding boxes}. Our extensive experiments show that our approach scales well and outperforms the state of the art by up to 3 orders of magnitude w.r.t. to its runtime. Moreover, we show that our approach to discarding computation of topological relations is more effective than the state of the art and leads to less computations of topological relations having to be carried out.
The contributions of this paper can be summarized as follows: 
(1) We present a novel indexing algorithm for geo-spatial resources based on an optimized sparse space tiling. (2) We provide a novel filtering approach for the rapid discovery of topological relations (\radon), which uses minimum bounding box (MBB) approximation. (3) We show that \radon is able to discover any of the  DE-9IM relations that involve intersection of at least one point.
(4) We evaluate \radon on real datasets and show 
that it clearly outperforms the state of the art.




\section{Preliminaries}
\label{sec:preliminaries}



Let $K$ be a finite RDF knowledge base.
$K$ can be regarded as a set of triples $(s, p, o) \in (\mathcal{R}  \cup \mathcal{B}) \times \mathcal{P} \times (\mathcal{R} \cup \mathcal{L} \cup \mathcal{B})$, where $\mathcal{R}$ is the set of all resources, $\mathcal{B}$ is the set of all blank nodes, $\mathcal{P}$ the set of all predicates and $\mathcal{L}$ the set of all literals.
Given a set of source resources $S$ and target resources $T$ from two (not necessarily distinct) knowledge bases $K_1$ and $K_2$
as well as a relation $R$, the goal of \emph{Link Discovery} (LD) is is to find the set of \emph{mapping} $M = \{(s,t) \in S \times T: R(s,t)\}$.
Naive computation of $M$ requires quadratic time complexity to compare every $s \in S$ with every $t \in T$, which is clearly impracticable for large datasets such as geo-spatial datasets, which are the focus of this work. 
Here, we present an algorithm for efficient computations of topological relations between resources with geo-spatial descriptions (i.e., described by means of vector geometry).\footnote{Most commonly encoded in the WKT format, see \url{http://www.opengeospatial.org/standards/sfa}.} 
We assume that each of the resources in $S$ and $T$ considered in the subsequent portion of this paper as being described by a geometry, where each geometry is modelled as sequence of points. 
An example of such resources is shown in \autoref{fig:leipzig1}.

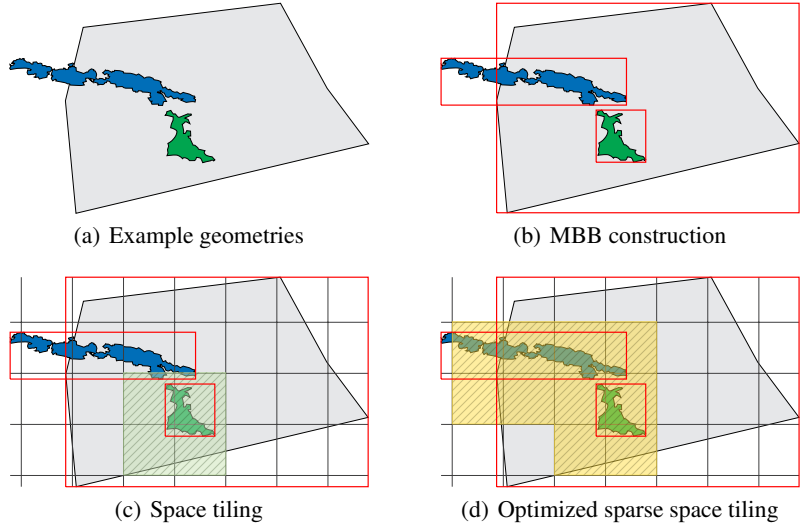
\begin{figure}[t]
\centering
\subfigure[Example geometries]{\label{fig:leipzig1}
    \resizebox{0.45\linewidth}{!}{
        \begin{tikzpicture}[scale=3]
        \draw[fill=Gray!20!white, line width=0.01pt] (12.24334845000007, 51.348402) -- (12.260675450000065, 51.42005750000007) -- (12.45318745000003, 51.44363850000005) -- (12.498737950000077, 51.36050450000005) -- (12.5391649500001, 51.3063565) -- (12.253649450000097, 51.23853550000001) -- (12.24334845000007, 51.348402);

        \draw[fill=Green, line width=0.01pt]
        (12.350457680302574, 51.335741860486834) -- (12.352140928324896, 51.333637635911934) -- (12.352439235513609, 51.333260718584498) -- (12.35313583488478, 51.332573109900757) -- (12.353698084666826, 51.332288104893777) -- (12.354745816406369, 51.332356893302503) -- (12.357841093431484, 51.333607855296478) -- (12.362974379474052, 51.335692068272955) -- (12.363614519530703, 51.335366147209989) -- (12.365929215946238, 51.335501366978214) -- (12.366311547463905, 51.334507614038031) -- (12.364308429206481, 51.333575625095392) -- (12.359845102391926, 51.331504636298114) -- (12.363055148685051, 51.33018700453627) -- (12.364807826900106, 51.325985154216667) -- (12.361284026250093, 51.327881992318808) -- (12.359637477398138, 51.327576574910601) -- (12.358879859135268, 51.326425098612404) -- (12.357259861513205, 51.325364351195581) -- (12.357642647938114, 51.324289466492395) -- (12.357946704334694, 51.322987080911844) -- (12.358780590768848, 51.318658908839012) -- (12.358634489819959, 51.316495595865717) -- (12.360475078143981, 51.316603469281716) -- (12.3617195173801, 51.319237258521404) -- (12.36541261590396, 51.318790499281072) -- (12.365630387854567, 51.313999496266284) -- (12.362997395756869, 51.313763242389527) -- (12.363995000516482, 51.311899434919425) -- (12.365051262574289, 51.311874076620185) -- (12.365607273029472, 51.310066772254423) -- (12.367888419286293, 51.305096585224398) -- (12.369674573591915, 51.304967455872415) -- (12.372537273655375, 51.302571373158798) -- (12.373086928998044, 51.301165079604004) -- (12.379545575204549, 51.299666971379082) -- (12.37872328759245, 51.298212804499961) -- (12.378921390603191, 51.29727894345843) -- (12.379763225230846, 51.295991952475191) -- (12.381119015331068, 51.294987706263804) -- (12.383924177183824, 51.294093603416748) -- (12.386512362349434, 51.293660424501084) -- (12.389192648396918, 51.290668357656699) -- (12.387995633329179, 51.289668118388406) -- (12.387473747627904, 51.290413888067413) -- (12.385350791642892, 51.290262653493834) -- (12.385057800671893, 51.289804036134811) -- (12.381065358211419, 51.288987702341863) -- (12.377334617220791, 51.287971108068191) -- (12.375926843008681, 51.291272742814236) -- (12.374910342900066, 51.292508760459086) -- (12.37300172153798, 51.292955671587016) -- (12.371290236206972, 51.293153338091798) -- (12.370511859464582, 51.291110323034957) -- (12.369462259533334, 51.290022259850353) -- (12.36454034937703, 51.289959109934649) -- (12.362420974717045, 51.291903825165875) -- (12.360083167437903, 51.292796369108409) -- (12.358273361117101, 51.292718268534998) -- (12.357712950055422, 51.291838759381911) -- (12.355791296729445, 51.293606601141796) -- (12.354171523467116, 51.294559231622756) -- (12.352414308467877, 51.295301425808766) -- (12.350630974187782, 51.295432018276152) -- (12.348374522962821, 51.295403685406008) -- (12.342410999386157, 51.294801344483417) -- (12.343446198950682, 51.297354061499775) -- (12.342943438819521, 51.30071634443685) -- (12.344034825445529, 51.304040806650157) -- (12.343555532700668, 51.306082737522757) -- (12.343278938466336, 51.306738587359575) -- (12.345592045420219, 51.308126661912382) -- (12.346437286017204, 51.308793267939492) -- (12.347353938594205, 51.309508555988543) -- (12.348880224960007, 51.311222686411064) -- (12.350055199360009, 51.312683311402537) -- (12.351202153992402, 51.314568733611019) -- (12.351935129380868, 51.316675396602967) -- (12.352188800182851, 51.318364204327793) -- (12.351940213359418, 51.318853258844456) -- (12.351641954959257, 51.31944097097989) -- (12.350388470599748, 51.320347897807686) -- (12.348172259629171, 51.321417335084192) -- (12.345990332673312, 51.322030266067316) -- (12.346490708957655, 51.324833275483222) -- (12.347353460398706, 51.327208587621207) -- (12.34836525703229, 51.329421467505341) -- (12.348559679448288, 51.332634488908042) -- (12.344419357022041, 51.333669329491201) -- (12.342416238561107, 51.334167976991068) -- (12.340703846780286, 51.334227278551197) -- (12.341664189855061, 51.335511184311514) -- (12.342043977133882, 51.336022046304308) -- (12.
342270644067744, 51.337729514149935) -- (12.342342068056602, 51.337755856622543) -- (12.346176128477437, 51.339026334031388) -- (12.346268259457363, 51.338851763161273) -- (12.347353684019337, 51.33710045182044) -- (12.348738598565443, 51.335337963082502) -- (12.350457680302574, 51.335741860486834);

        \draw[fill=NavyBlue, line width=0.01pt] (12.246522753244609, 51.384075217037221) -- (12.247393832331706, 51.384372062492751) -- (12.249163829478919, 51.384785924286867) -- (12.249415284430482, 51.384768019121104) -- (12.250445526525292, 51.385298359110479) -- (12.25147057030267, 51.384750516173575) -- (12.251535098380623, 51.384719085433908) -- (12.249654309392385, 51.382077249513983) -- (12.250816656546887, 51.381651442548353) -- (12.252197359372134, 51.381849688619759) -- (12.25271828098597, 51.381020470191942) -- (12.253766194393881, 51.380730110400258) -- (12.253888913368309, 51.377456252920126) -- (12.259271130637494, 51.375722038656932) -- (12.260885093595359, 51.375061607050284) -- (12.262352700674414, 51.375576271034383) -- (12.264159634897592, 51.376532200546407) -- (12.265236264932071, 51.376311137787653) -- (12.266335283287249, 51.378085757914015) -- (12.268697296117351, 51.379555188036996) -- (12.27108572302801, 51.378798948027466) -- (12.270887879624965, 51.377845254013053) -- (12.272487634123451, 51.377432333310985) -- (12.273498662599771, 51.375288410183721) -- (12.27380881647413, 51.373826273197267) -- (12.27313748501723, 51.371055063531912) -- (12.273277240542997, 51.370043659351246) -- (12.274438352045674, 51.36950772665967) -- (12.276010701822814, 51.36952599330121) -- (12.276695558642889, 51.370266910024725) -- (12.27920018445908, 51.370102219271558) -- (12.280096156493189, 51.369786586921151) -- (12.280868623471511, 51.369106962205244) -- (12.282152852820929, 51.369003327444418) -- (12.282909695584273, 51.369832290103048) -- (12.282583808915742, 51.370885012220768) -- (12.285480368612724, 51.370253988591884) -- (12.286542245745107, 51.371971310788872) -- (12.284743954871489, 51.372654544283556) -- (12.285617613345947, 51.373942359654343) -- (12.283289090612593, 51.374879011447234) -- (12.281813680648778, 51.373893277483816) -- (12.279653525205013, 51.375761525447771) -- (12.280236922803141, 51.376232879931557) -- (12.280870718316358, 51.376739875609651) -- (12.282016629140035, 51.377834282255868) -- (12.282428735362069, 51.378533158196767) -- (12.283810966985126, 51.377810793806972) -- (12.284551938088933, 51.378210890321448) -- (12.285447715690598, 51.377893797951408) -- (12.28726608794922, 51.37685398694606) -- (12.288633682489589, 51.376161590354052) -- (12.289334172401844, 51.377193200583164) -- (12.290954638092328, 51.376940291064393) -- (12.29130423239611, 51.378115618172401) -- (12.296281390888996, 51.377762572473443) -- (12.29969187538671, 51.376950178516097) -- (12.299987086197998, 51.375641304176433) -- (12.302524948741569, 51.375351667977668) -- (12.306842510490204, 51.37525786524693) -- (12.310553507338504, 51.374333876474012) -- (12.312428592008983, 51.373491247994451) -- (12.314735468685871, 51.374610329100456) -- (12.319036580904704, 51.371726907262691) -- (12.320019886486373, 51.370566273950537) -- (12.321066342504787, 51.370727793677368) -- (12.322451702406559, 51.371437542380377) -- (12.324263967567594, 51.371636206215392) -- (12.325831351407274, 51.37160326616663) -- (12.32882323882083, 51.370348891019866) -- (12.329632057738772, 51.368982837480445) -- (12.333575239733351, 51.368290860162546) -- (12.335622230574739, 51.367495138798176) -- (12.336847136310912, 51.366712047608736) -- (12.336000701166544, 51.366001475159202) -- (12.338461789447825, 51.360832288791457) -- (12.341129193123255, 51.36186309942056) -- (12.34195650811335, 51.360224780975727) -- (12.344341513596744, 51.359113696037994) -- (12.345136446831756, 51.359787829310676) -- (12.346735493632552, 51.358573530505737) -- (12.347600792737344, 51.357568696531743) -- (12.35015624313213, 51.35671011433756) -- (12.352275044928382, 51.35711708078712) -- (12.355485947000082, 51.354999292113362) -- (12.35832209243811, 51.35662900342497) -- (12.359445439245739, 51.35618882106079) -- (12.362241222487967, 51.355492279136996) -- (12.364495716724051, 51.354852519015076) -- (12.365319549672318, 51.354396114121855) -- (12.367677737660886, 51.353165934808551) -- (12.369392736796422, 51.352926756915281) -- (12.36984169500715, 51.
351503681367561) -- (12.368564823621767, 51.351003258007282) -- (12.369223947483182, 51.350311194184769) -- (12.370352998643272, 51.349606977605731) -- (12.369851845600937, 51.347190620284501) -- (12.368009434018132, 51.347609557387749) -- (12.365737688287201, 51.34748229015446) -- (12.361164044798107, 51.346848302016092) -- (12.355994746647676, 51.349020015171874) -- (12.357658987482621, 51.350182885602777) -- (12.354004340495889, 51.350090335631627) -- (12.353282344220446, 51.351255093780189) -- (12.348831933359381, 51.350354508052533) -- (12.346311976552157, 51.353051007964702) -- (12.344517853737383, 51.35199511390072) -- (12.343404339602083, 51.351765068543145) -- (12.341345629308595, 51.352508274911536) -- (12.341076677906312, 51.353058477379868) -- (12.340291658571687, 51.354798218401335) -- (12.33977537650968, 51.355934269469216) -- (12.338198657694754, 51.356107697058178) -- (12.336758592041136, 51.358323816021944) -- (12.335715219317326, 51.358146109064059) -- (12.33651441206519, 51.356277026331561) -- (12.33353315389297, 51.357313588586649) -- (12.333122373602066, 51.356436679959451) -- (12.33723804256798, 51.355135645118878) -- (12.338312617804823, 51.354295765937707) -- (12.342005117713381, 51.347700461785969) -- (12.339771743403764, 51.347659560881731) -- (12.339021377452122, 51.34736935182292) -- (12.338145351555761, 51.343882052985677) -- (12.334104078488531, 51.343939429636244) -- (12.334208436228492, 51.345165267445587) -- (12.327811416678017, 51.345031704900158) -- (12.326931383170262, 51.346895114437444) -- (12.327814971934759, 51.34795678434886) -- (12.326212706909454, 51.347923861387216) -- (12.323977965678798, 51.351117154271265) -- (12.32599683224743, 51.351411765920325) -- (12.326634728325402, 51.351743415338447) -- (12.326603265275594, 51.352876257338892) -- (12.325826521488141, 51.353651145174013) -- (12.326372348581211, 51.356043841924567) -- (12.323706028339783, 51.356692835029129) -- (12.318572767864996, 51.356935508239829) -- (12.318365016232223, 51.356940722877987) -- (12.317838545365285, 51.357879263415711) -- (12.314636157536155, 51.358260729602989) -- (12.313942772650453, 51.356979426343052) -- (12.311870153345206, 51.356847496509239) -- (12.311741000486265, 51.356825949081063) -- (12.311414890693008, 51.357763935971192) -- (12.309736525560789, 51.357667320822657) -- (12.308901940797451, 51.358782565020512) -- (12.30860839859281, 51.360683275497514) -- (12.302067429369602, 51.36072216101087) -- (12.300204670240651, 51.361253211641476) -- (12.29897523894412, 51.362295542285459) -- (12.294893281479373, 51.362134065384943) -- (12.292976647459856, 51.361755851946569) -- (12.29247888874972, 51.364507023295516) -- (12.291283218060554, 51.365002892837012) -- (12.289297321331267, 51.365190711121755) -- (12.287333893244178, 51.365643213027894) -- (12.286573135568434, 51.367297759182676) -- (12.283978188948126, 51.36787279504874) -- (12.283318190508282, 51.366117002088174) -- (12.281746586784671, 51.36608473512895) -- (12.281477511731753, 51.365176116131728) -- (12.28039098594331, 51.364428798439732) -- (12.279227603716308, 51.364217859264741) -- (12.277642273343293, 51.364427907161243) -- (12.274905197989451, 51.363822518882955) -- (12.275001952726143, 51.364743700392999) -- (12.275091059443319, 51.365536595874865) -- (12.271440587291558, 51.366593565979514) -- (12.270422662707242, 51.367062062308619) -- (12.270489858935029, 51.367762191973824) -- (12.270922902640701, 51.368356553461652) -- (12.272059656836335, 51.368951737505981) -- (12.26937012395509, 51.369848641217544) -- (12.267475676358472, 51.370271444919716) -- (12.269696899969295, 51.366928260472775) -- (12.267004223809117, 51.369020067211252) -- (12.265774510392861, 51.370565570680945) -- (12.263152048553426, 51.370404783146533) -- (12.26302620893426, 51.369384238102391) -- (12.263816131530968, 51.367311711123655) -- (12.260364409727808, 51.366292946190796) -- (12.261799952116782, 51.363857880878328) -- (12.259521110551093, 51.363386894067759) -- (12.257257735681563, 51.363300098103672) -- (12.248183285952161, 51.36373488862143) -- (
12.247690183630381, 51.364600939469582) -- (12.246814861646913, 51.36507150134041) -- (12.245807838577376, 51.364836957468761) -- (12.2438414232022, 51.365773327486622) -- (12.242426683808604, 51.366180294147348) -- (12.233934642339268, 51.367788272869532) -- (12.234484927288017, 51.368934273987563) -- (12.233465776278955, 51.369447884088252) -- (12.230922483090223, 51.367880553669963) -- (12.22953548944434, 51.370177079255811) -- (12.230099422680977, 51.37105808334271) -- (12.229842750557417, 51.371660584527703) -- (12.231052857454873, 51.372262195645582) -- (12.230884004297474, 51.373259069857738) -- (12.229302474513062, 51.373332142595928) -- (12.230802855177876, 51.378160297483859) -- (12.229696317568656, 51.378334826817017) -- (12.228186214683037, 51.376558596519125) -- (12.227115492357047, 51.376714832637248) -- (12.225691352683789, 51.376605823482862) -- (12.225578943352197, 51.375464919097098) -- (12.227448529904082, 51.375489440295496) -- (12.227565725248423, 51.374193644304285) -- (12.226155072999806, 51.373837495877531) -- (12.225410889648259, 51.372882074620115) -- (12.225348309546142, 51.371494457069495) -- (12.223764747762653, 51.371017158467502) -- (12.221369018889131, 51.370632276045576) -- (12.218067307565756, 51.369682497044671) -- (12.217878274688685, 51.369628115558939) -- (12.216799331652942, 51.371601486104396) -- (12.214367642473924, 51.371217582450072) -- (12.213139333919457, 51.371575158867969) -- (12.214399114541607, 51.371993547359359) -- (12.214432130483928, 51.373183108670702) -- (12.214825602475132, 51.374585306376218) -- (12.21352525238807, 51.374895045775048) -- (12.212472506513214, 51.374112315159834) -- (12.210920030681038, 51.374300700937731) -- (12.208855089210765, 51.373954015339706) -- (12.207395123716301, 51.373868363441403) -- (12.206677229822247, 51.37419202589119) -- (12.205979325200865, 51.376071468299806) -- (12.206664951934053, 51.37864800143619) -- (12.209164683146307, 51.377915459998064) -- (12.21003067526938, 51.378693549586004) -- (12.208660226043499, 51.379480155630617) -- (12.207316992164763, 51.379880122418115) -- (12.205355005800122, 51.380217142385057) -- (12.203485523737612, 51.380306993062355) -- (12.196573449841249, 51.380038137495852) -- (12.189680711363202, 51.378930685215899) -- (12.188973074299239, 51.380239222226272) -- (12.190193915401418, 51.381888291124341) -- (12.189592515705305, 51.382675243390416) -- (12.189272732398623, 51.383785672531978) -- (12.190384481971387, 51.382922087469261) -- (12.192672664883768, 51.383231299206834) -- (12.193606339068129, 51.38480885197319) -- (12.190887846164504, 51.384904201159678) -- (12.19158950649085, 51.385969259563645) -- (12.192863426489581, 51.386162137524138) -- (12.192311861490763, 51.386764643605304) -- (12.192998436574038, 51.387582547443245) -- (12.194429852690803, 51.387550753040934) -- (12.19494056629668, 51.387541663233883) -- (12.195577605755629, 51.388517024151348) -- (12.196484968817158, 51.388737181490043) -- (12.198515207773884, 51.389092182978622) -- (12.198920615257807, 51.389676584445702) -- (12.200068012383136, 51.388835931904715) -- (12.200945775506133, 51.388858438214989) -- (12.201344427512213, 51.389532811064811) -- (12.202510721082682, 51.389721734104974) -- (12.203342958427397, 51.3892183046855) -- (12.202590734446687, 51.388305569265626) -- (12.204262696533192, 51.387366094075006) -- (12.205623761751248, 51.386075892373796) -- (12.206976600504332, 51.386202012809129) -- (12.208198227474099, 51.385936980417334) -- (12.21058531553262, 51.384103026983574) -- (12.211701154056778, 51.384337165409981) -- (12.211780736382146, 51.386050344368329) -- (12.215134025612725, 51.386388744007398) -- (12.216442602397949, 51.384734555916474) -- (12.218911314728679, 51.385140302577859) -- (12.220254642979143, 51.384808187210901) -- (12.220820148197198, 51.382650182034233) -- (12.223481853697647, 51.382858396841108) -- (12.223767471441324, 51.382408927078131) -- (12.225067750141895, 51.382121896291295) -- (12.225287937646275, 51.381510548831692) -- (12.226181343522086, 51.381902228068697) -- (12.227799005632077,
 51.381849234565024) -- (12.229384499486233, 51.381045459890373) -- (12.229813750119721, 51.380555707046931) -- (12.229874475161107, 51.379643142340882) -- (12.230117195541764, 51.379238710595864) -- (12.231849210872458, 51.379055597880139) -- (12.232584834407923, 51.379577555532812) -- (12.23302526762907, 51.378332801639495) -- (12.235542600598729, 51.378559573089518) -- (12.237263495397762, 51.379176422542287) -- (12.238939239387522, 51.379335028002657) -- (12.24180027646884, 51.379315207596647) -- (12.241696926658449, 51.380343292881619) -- (12.243075743052062, 51.380034492603926) -- (12.243473315602097, 51.378881899864723) -- (12.246480180663648, 51.379418772581886) -- (12.246393296626321, 51.380921424997098) -- (12.248846816864127, 51.38151110727091) -- (12.247146944640511, 51.382287280461576) -- (12.246023435322801, 51.38335612324277) -- (12.246097988045527, 51.383926783297113) -- (12.246522753244609, 51.384075217037221);

        \end{tikzpicture}
    }
}~\subfigure[MBB construction]{\label{fig:leipzig2}
    \resizebox{0.45\linewidth}{!}{
        \begin{tikzpicture}[scale=3]

        \draw[fill=Gray!20!white, line width=0.01pt] (12.24334845000007, 51.348402) -- (12.260675450000065, 51.42005750000007) -- (12.45318745000003, 51.44363850000005) -- (12.498737950000077, 51.36050450000005) -- (12.5391649500001, 51.3063565) -- (12.253649450000097, 51.23853550000001) -- (12.24334845000007, 51.348402);

        \draw[fill=Green, line width=0.01pt]
        (12.350457680302574, 51.335741860486834) -- (12.352140928324896, 51.333637635911934) -- (12.352439235513609, 51.333260718584498) -- (12.35313583488478, 51.332573109900757) -- (12.353698084666826, 51.332288104893777) -- (12.354745816406369, 51.332356893302503) -- (12.357841093431484, 51.333607855296478) -- (12.362974379474052, 51.335692068272955) -- (12.363614519530703, 51.335366147209989) -- (12.365929215946238, 51.335501366978214) -- (12.366311547463905, 51.334507614038031) -- (12.364308429206481, 51.333575625095392) -- (12.359845102391926, 51.331504636298114) -- (12.363055148685051, 51.33018700453627) -- (12.364807826900106, 51.325985154216667) -- (12.361284026250093, 51.327881992318808) -- (12.359637477398138, 51.327576574910601) -- (12.358879859135268, 51.326425098612404) -- (12.357259861513205, 51.325364351195581) -- (12.357642647938114, 51.324289466492395) -- (12.357946704334694, 51.322987080911844) -- (12.358780590768848, 51.318658908839012) -- (12.358634489819959, 51.316495595865717) -- (12.360475078143981, 51.316603469281716) -- (12.3617195173801, 51.319237258521404) -- (12.36541261590396, 51.318790499281072) -- (12.365630387854567, 51.313999496266284) -- (12.362997395756869, 51.313763242389527) -- (12.363995000516482, 51.311899434919425) -- (12.365051262574289, 51.311874076620185) -- (12.365607273029472, 51.310066772254423) -- (12.367888419286293, 51.305096585224398) -- (12.369674573591915, 51.304967455872415) -- (12.372537273655375, 51.302571373158798) -- (12.373086928998044, 51.301165079604004) -- (12.379545575204549, 51.299666971379082) -- (12.37872328759245, 51.298212804499961) -- (12.378921390603191, 51.29727894345843) -- (12.379763225230846, 51.295991952475191) -- (12.381119015331068, 51.294987706263804) -- (12.383924177183824, 51.294093603416748) -- (12.386512362349434, 51.293660424501084) -- (12.389192648396918, 51.290668357656699) -- (12.387995633329179, 51.289668118388406) -- (12.387473747627904, 51.290413888067413) -- (12.385350791642892, 51.290262653493834) -- (12.385057800671893, 51.289804036134811) -- (12.381065358211419, 51.288987702341863) -- (12.377334617220791, 51.287971108068191) -- (12.375926843008681, 51.291272742814236) -- (12.374910342900066, 51.292508760459086) -- (12.37300172153798, 51.292955671587016) -- (12.371290236206972, 51.293153338091798) -- (12.370511859464582, 51.291110323034957) -- (12.369462259533334, 51.290022259850353) -- (12.36454034937703, 51.289959109934649) -- (12.362420974717045, 51.291903825165875) -- (12.360083167437903, 51.292796369108409) -- (12.358273361117101, 51.292718268534998) -- (12.357712950055422, 51.291838759381911) -- (12.355791296729445, 51.293606601141796) -- (12.354171523467116, 51.294559231622756) -- (12.352414308467877, 51.295301425808766) -- (12.350630974187782, 51.295432018276152) -- (12.348374522962821, 51.295403685406008) -- (12.342410999386157, 51.294801344483417) -- (12.343446198950682, 51.297354061499775) -- (12.342943438819521, 51.30071634443685) -- (12.344034825445529, 51.304040806650157) -- (12.343555532700668, 51.306082737522757) -- (12.343278938466336, 51.306738587359575) -- (12.345592045420219, 51.308126661912382) -- (12.346437286017204, 51.308793267939492) -- (12.347353938594205, 51.309508555988543) -- (12.348880224960007, 51.311222686411064) -- (12.350055199360009, 51.312683311402537) -- (12.351202153992402, 51.314568733611019) -- (12.351935129380868, 51.316675396602967) -- (12.352188800182851, 51.318364204327793) -- (12.351940213359418, 51.318853258844456) -- (12.351641954959257, 51.31944097097989) -- (12.350388470599748, 51.320347897807686) -- (12.348172259629171, 51.321417335084192) -- (12.345990332673312, 51.322030266067316) -- (12.346490708957655, 51.324833275483222) -- (12.347353460398706, 51.327208587621207) -- (12.34836525703229, 51.329421467505341) -- (12.348559679448288, 51.332634488908042) -- (12.344419357022041, 51.333669329491201) -- (12.342416238561107, 51.334167976991068) -- (12.340703846780286, 51.334227278551197) -- (12.341664189855061, 51.335511184311514) -- (12.342043977133882, 51.336022046304308) -- (12.
342270644067744, 51.337729514149935) -- (12.342342068056602, 51.337755856622543) -- (12.346176128477437, 51.339026334031388) -- (12.346268259457363, 51.338851763161273) -- (12.347353684019337, 51.33710045182044) -- (12.348738598565443, 51.335337963082502) -- (12.350457680302574, 51.335741860486834);

        \draw[fill=NavyBlue, line width=0.01pt] (12.246522753244609, 51.384075217037221) -- (12.247393832331706, 51.384372062492751) -- (12.249163829478919, 51.384785924286867) -- (12.249415284430482, 51.384768019121104) -- (12.250445526525292, 51.385298359110479) -- (12.25147057030267, 51.384750516173575) -- (12.251535098380623, 51.384719085433908) -- (12.249654309392385, 51.382077249513983) -- (12.250816656546887, 51.381651442548353) -- (12.252197359372134, 51.381849688619759) -- (12.25271828098597, 51.381020470191942) -- (12.253766194393881, 51.380730110400258) -- (12.253888913368309, 51.377456252920126) -- (12.259271130637494, 51.375722038656932) -- (12.260885093595359, 51.375061607050284) -- (12.262352700674414, 51.375576271034383) -- (12.264159634897592, 51.376532200546407) -- (12.265236264932071, 51.376311137787653) -- (12.266335283287249, 51.378085757914015) -- (12.268697296117351, 51.379555188036996) -- (12.27108572302801, 51.378798948027466) -- (12.270887879624965, 51.377845254013053) -- (12.272487634123451, 51.377432333310985) -- (12.273498662599771, 51.375288410183721) -- (12.27380881647413, 51.373826273197267) -- (12.27313748501723, 51.371055063531912) -- (12.273277240542997, 51.370043659351246) -- (12.274438352045674, 51.36950772665967) -- (12.276010701822814, 51.36952599330121) -- (12.276695558642889, 51.370266910024725) -- (12.27920018445908, 51.370102219271558) -- (12.280096156493189, 51.369786586921151) -- (12.280868623471511, 51.369106962205244) -- (12.282152852820929, 51.369003327444418) -- (12.282909695584273, 51.369832290103048) -- (12.282583808915742, 51.370885012220768) -- (12.285480368612724, 51.370253988591884) -- (12.286542245745107, 51.371971310788872) -- (12.284743954871489, 51.372654544283556) -- (12.285617613345947, 51.373942359654343) -- (12.283289090612593, 51.374879011447234) -- (12.281813680648778, 51.373893277483816) -- (12.279653525205013, 51.375761525447771) -- (12.280236922803141, 51.376232879931557) -- (12.280870718316358, 51.376739875609651) -- (12.282016629140035, 51.377834282255868) -- (12.282428735362069, 51.378533158196767) -- (12.283810966985126, 51.377810793806972) -- (12.284551938088933, 51.378210890321448) -- (12.285447715690598, 51.377893797951408) -- (12.28726608794922, 51.37685398694606) -- (12.288633682489589, 51.376161590354052) -- (12.289334172401844, 51.377193200583164) -- (12.290954638092328, 51.376940291064393) -- (12.29130423239611, 51.378115618172401) -- (12.296281390888996, 51.377762572473443) -- (12.29969187538671, 51.376950178516097) -- (12.299987086197998, 51.375641304176433) -- (12.302524948741569, 51.375351667977668) -- (12.306842510490204, 51.37525786524693) -- (12.310553507338504, 51.374333876474012) -- (12.312428592008983, 51.373491247994451) -- (12.314735468685871, 51.374610329100456) -- (12.319036580904704, 51.371726907262691) -- (12.320019886486373, 51.370566273950537) -- (12.321066342504787, 51.370727793677368) -- (12.322451702406559, 51.371437542380377) -- (12.324263967567594, 51.371636206215392) -- (12.325831351407274, 51.37160326616663) -- (12.32882323882083, 51.370348891019866) -- (12.329632057738772, 51.368982837480445) -- (12.333575239733351, 51.368290860162546) -- (12.335622230574739, 51.367495138798176) -- (12.336847136310912, 51.366712047608736) -- (12.336000701166544, 51.366001475159202) -- (12.338461789447825, 51.360832288791457) -- (12.341129193123255, 51.36186309942056) -- (12.34195650811335, 51.360224780975727) -- (12.344341513596744, 51.359113696037994) -- (12.345136446831756, 51.359787829310676) -- (12.346735493632552, 51.358573530505737) -- (12.347600792737344, 51.357568696531743) -- (12.35015624313213, 51.35671011433756) -- (12.352275044928382, 51.35711708078712) -- (12.355485947000082, 51.354999292113362) -- (12.35832209243811, 51.35662900342497) -- (12.359445439245739, 51.35618882106079) -- (12.362241222487967, 51.355492279136996) -- (12.364495716724051, 51.354852519015076) -- (12.365319549672318, 51.354396114121855) -- (12.367677737660886, 51.353165934808551) -- (12.369392736796422, 51.352926756915281) -- (12.36984169500715, 51.
351503681367561) -- (12.368564823621767, 51.351003258007282) -- (12.369223947483182, 51.350311194184769) -- (12.370352998643272, 51.349606977605731) -- (12.369851845600937, 51.347190620284501) -- (12.368009434018132, 51.347609557387749) -- (12.365737688287201, 51.34748229015446) -- (12.361164044798107, 51.346848302016092) -- (12.355994746647676, 51.349020015171874) -- (12.357658987482621, 51.350182885602777) -- (12.354004340495889, 51.350090335631627) -- (12.353282344220446, 51.351255093780189) -- (12.348831933359381, 51.350354508052533) -- (12.346311976552157, 51.353051007964702) -- (12.344517853737383, 51.35199511390072) -- (12.343404339602083, 51.351765068543145) -- (12.341345629308595, 51.352508274911536) -- (12.341076677906312, 51.353058477379868) -- (12.340291658571687, 51.354798218401335) -- (12.33977537650968, 51.355934269469216) -- (12.338198657694754, 51.356107697058178) -- (12.336758592041136, 51.358323816021944) -- (12.335715219317326, 51.358146109064059) -- (12.33651441206519, 51.356277026331561) -- (12.33353315389297, 51.357313588586649) -- (12.333122373602066, 51.356436679959451) -- (12.33723804256798, 51.355135645118878) -- (12.338312617804823, 51.354295765937707) -- (12.342005117713381, 51.347700461785969) -- (12.339771743403764, 51.347659560881731) -- (12.339021377452122, 51.34736935182292) -- (12.338145351555761, 51.343882052985677) -- (12.334104078488531, 51.343939429636244) -- (12.334208436228492, 51.345165267445587) -- (12.327811416678017, 51.345031704900158) -- (12.326931383170262, 51.346895114437444) -- (12.327814971934759, 51.34795678434886) -- (12.326212706909454, 51.347923861387216) -- (12.323977965678798, 51.351117154271265) -- (12.32599683224743, 51.351411765920325) -- (12.326634728325402, 51.351743415338447) -- (12.326603265275594, 51.352876257338892) -- (12.325826521488141, 51.353651145174013) -- (12.326372348581211, 51.356043841924567) -- (12.323706028339783, 51.356692835029129) -- (12.318572767864996, 51.356935508239829) -- (12.318365016232223, 51.356940722877987) -- (12.317838545365285, 51.357879263415711) -- (12.314636157536155, 51.358260729602989) -- (12.313942772650453, 51.356979426343052) -- (12.311870153345206, 51.356847496509239) -- (12.311741000486265, 51.356825949081063) -- (12.311414890693008, 51.357763935971192) -- (12.309736525560789, 51.357667320822657) -- (12.308901940797451, 51.358782565020512) -- (12.30860839859281, 51.360683275497514) -- (12.302067429369602, 51.36072216101087) -- (12.300204670240651, 51.361253211641476) -- (12.29897523894412, 51.362295542285459) -- (12.294893281479373, 51.362134065384943) -- (12.292976647459856, 51.361755851946569) -- (12.29247888874972, 51.364507023295516) -- (12.291283218060554, 51.365002892837012) -- (12.289297321331267, 51.365190711121755) -- (12.287333893244178, 51.365643213027894) -- (12.286573135568434, 51.367297759182676) -- (12.283978188948126, 51.36787279504874) -- (12.283318190508282, 51.366117002088174) -- (12.281746586784671, 51.36608473512895) -- (12.281477511731753, 51.365176116131728) -- (12.28039098594331, 51.364428798439732) -- (12.279227603716308, 51.364217859264741) -- (12.277642273343293, 51.364427907161243) -- (12.274905197989451, 51.363822518882955) -- (12.275001952726143, 51.364743700392999) -- (12.275091059443319, 51.365536595874865) -- (12.271440587291558, 51.366593565979514) -- (12.270422662707242, 51.367062062308619) -- (12.270489858935029, 51.367762191973824) -- (12.270922902640701, 51.368356553461652) -- (12.272059656836335, 51.368951737505981) -- (12.26937012395509, 51.369848641217544) -- (12.267475676358472, 51.370271444919716) -- (12.269696899969295, 51.366928260472775) -- (12.267004223809117, 51.369020067211252) -- (12.265774510392861, 51.370565570680945) -- (12.263152048553426, 51.370404783146533) -- (12.26302620893426, 51.369384238102391) -- (12.263816131530968, 51.367311711123655) -- (12.260364409727808, 51.366292946190796) -- (12.261799952116782, 51.363857880878328) -- (12.259521110551093, 51.363386894067759) -- (12.257257735681563, 51.363300098103672) -- (12.248183285952161, 51.36373488862143) -- (
12.247690183630381, 51.364600939469582) -- (12.246814861646913, 51.36507150134041) -- (12.245807838577376, 51.364836957468761) -- (12.2438414232022, 51.365773327486622) -- (12.242426683808604, 51.366180294147348) -- (12.233934642339268, 51.367788272869532) -- (12.234484927288017, 51.368934273987563) -- (12.233465776278955, 51.369447884088252) -- (12.230922483090223, 51.367880553669963) -- (12.22953548944434, 51.370177079255811) -- (12.230099422680977, 51.37105808334271) -- (12.229842750557417, 51.371660584527703) -- (12.231052857454873, 51.372262195645582) -- (12.230884004297474, 51.373259069857738) -- (12.229302474513062, 51.373332142595928) -- (12.230802855177876, 51.378160297483859) -- (12.229696317568656, 51.378334826817017) -- (12.228186214683037, 51.376558596519125) -- (12.227115492357047, 51.376714832637248) -- (12.225691352683789, 51.376605823482862) -- (12.225578943352197, 51.375464919097098) -- (12.227448529904082, 51.375489440295496) -- (12.227565725248423, 51.374193644304285) -- (12.226155072999806, 51.373837495877531) -- (12.225410889648259, 51.372882074620115) -- (12.225348309546142, 51.371494457069495) -- (12.223764747762653, 51.371017158467502) -- (12.221369018889131, 51.370632276045576) -- (12.218067307565756, 51.369682497044671) -- (12.217878274688685, 51.369628115558939) -- (12.216799331652942, 51.371601486104396) -- (12.214367642473924, 51.371217582450072) -- (12.213139333919457, 51.371575158867969) -- (12.214399114541607, 51.371993547359359) -- (12.214432130483928, 51.373183108670702) -- (12.214825602475132, 51.374585306376218) -- (12.21352525238807, 51.374895045775048) -- (12.212472506513214, 51.374112315159834) -- (12.210920030681038, 51.374300700937731) -- (12.208855089210765, 51.373954015339706) -- (12.207395123716301, 51.373868363441403) -- (12.206677229822247, 51.37419202589119) -- (12.205979325200865, 51.376071468299806) -- (12.206664951934053, 51.37864800143619) -- (12.209164683146307, 51.377915459998064) -- (12.21003067526938, 51.378693549586004) -- (12.208660226043499, 51.379480155630617) -- (12.207316992164763, 51.379880122418115) -- (12.205355005800122, 51.380217142385057) -- (12.203485523737612, 51.380306993062355) -- (12.196573449841249, 51.380038137495852) -- (12.189680711363202, 51.378930685215899) -- (12.188973074299239, 51.380239222226272) -- (12.190193915401418, 51.381888291124341) -- (12.189592515705305, 51.382675243390416) -- (12.189272732398623, 51.383785672531978) -- (12.190384481971387, 51.382922087469261) -- (12.192672664883768, 51.383231299206834) -- (12.193606339068129, 51.38480885197319) -- (12.190887846164504, 51.384904201159678) -- (12.19158950649085, 51.385969259563645) -- (12.192863426489581, 51.386162137524138) -- (12.192311861490763, 51.386764643605304) -- (12.192998436574038, 51.387582547443245) -- (12.194429852690803, 51.387550753040934) -- (12.19494056629668, 51.387541663233883) -- (12.195577605755629, 51.388517024151348) -- (12.196484968817158, 51.388737181490043) -- (12.198515207773884, 51.389092182978622) -- (12.198920615257807, 51.389676584445702) -- (12.200068012383136, 51.388835931904715) -- (12.200945775506133, 51.388858438214989) -- (12.201344427512213, 51.389532811064811) -- (12.202510721082682, 51.389721734104974) -- (12.203342958427397, 51.3892183046855) -- (12.202590734446687, 51.388305569265626) -- (12.204262696533192, 51.387366094075006) -- (12.205623761751248, 51.386075892373796) -- (12.206976600504332, 51.386202012809129) -- (12.208198227474099, 51.385936980417334) -- (12.21058531553262, 51.384103026983574) -- (12.211701154056778, 51.384337165409981) -- (12.211780736382146, 51.386050344368329) -- (12.215134025612725, 51.386388744007398) -- (12.216442602397949, 51.384734555916474) -- (12.218911314728679, 51.385140302577859) -- (12.220254642979143, 51.384808187210901) -- (12.220820148197198, 51.382650182034233) -- (12.223481853697647, 51.382858396841108) -- (12.223767471441324, 51.382408927078131) -- (12.225067750141895, 51.382121896291295) -- (12.225287937646275, 51.381510548831692) -- (12.226181343522086, 51.381902228068697) -- (12.227799005632077,
 51.381849234565024) -- (12.229384499486233, 51.381045459890373) -- (12.229813750119721, 51.380555707046931) -- (12.229874475161107, 51.379643142340882) -- (12.230117195541764, 51.379238710595864) -- (12.231849210872458, 51.379055597880139) -- (12.232584834407923, 51.379577555532812) -- (12.23302526762907, 51.378332801639495) -- (12.235542600598729, 51.378559573089518) -- (12.237263495397762, 51.379176422542287) -- (12.238939239387522, 51.379335028002657) -- (12.24180027646884, 51.379315207596647) -- (12.241696926658449, 51.380343292881619) -- (12.243075743052062, 51.380034492603926) -- (12.243473315602097, 51.378881899864723) -- (12.246480180663648, 51.379418772581886) -- (12.246393296626321, 51.380921424997098) -- (12.248846816864127, 51.38151110727091) -- (12.247146944640511, 51.382287280461576) -- (12.246023435322801, 51.38335612324277) -- (12.246097988045527, 51.383926783297113) -- (12.246522753244609, 51.384075217037221);
        
        
        \draw[red,ultra thin] (12.18897307429924, 51.34388205298568) -- (12.18897307429924, 51.389721734104974) -- (12.370352998643272, 51.389721734104974) -- (12.370352998643272, 51.34388205298568) -- (12.18897307429924, 51.34388205298568);
        
        \draw[red,ultra thin] (12.340703846780286, 51.28797110806819) -- (12.340703846780286, 51.33902633403139) -- (12.389192648396918, 51.33902633403139) -- (12.389192648396918, 51.28797110806819) -- (12.340703846780286, 51.28797110806819);

        \draw[red,ultra thin] (12.2433, 51.2385) --(12.2433, 51.4436) -- (12.5392, 51.4436) -- (12.5392, 51.2385) -- (12.2433, 51.2385);
        
        \end{tikzpicture}
    }
} \subfigure[Space tiling]{\label{fig:leipzig3}
    \resizebox{0.45\linewidth}{!}{
        \begin{tikzpicture}[scale=3]

        \draw[fill=Gray!20!white, line width=0.01pt] (12.24334845000007, 51.348402) -- (12.260675450000065, 51.42005750000007) -- (12.45318745000003, 51.44363850000005) -- (12.498737950000077, 51.36050450000005) -- (12.5391649500001, 51.3063565) -- (12.253649450000097, 51.23853550000001) -- (12.24334845000007, 51.348402);

        \draw[fill=Green, line width=0.01pt]
        (12.350457680302574, 51.335741860486834) -- (12.352140928324896, 51.333637635911934) -- (12.352439235513609, 51.333260718584498) -- (12.35313583488478, 51.332573109900757) -- (12.353698084666826, 51.332288104893777) -- (12.354745816406369, 51.332356893302503) -- (12.357841093431484, 51.333607855296478) -- (12.362974379474052, 51.335692068272955) -- (12.363614519530703, 51.335366147209989) -- (12.365929215946238, 51.335501366978214) -- (12.366311547463905, 51.334507614038031) -- (12.364308429206481, 51.333575625095392) -- (12.359845102391926, 51.331504636298114) -- (12.363055148685051, 51.33018700453627) -- (12.364807826900106, 51.325985154216667) -- (12.361284026250093, 51.327881992318808) -- (12.359637477398138, 51.327576574910601) -- (12.358879859135268, 51.326425098612404) -- (12.357259861513205, 51.325364351195581) -- (12.357642647938114, 51.324289466492395) -- (12.357946704334694, 51.322987080911844) -- (12.358780590768848, 51.318658908839012) -- (12.358634489819959, 51.316495595865717) -- (12.360475078143981, 51.316603469281716) -- (12.3617195173801, 51.319237258521404) -- (12.36541261590396, 51.318790499281072) -- (12.365630387854567, 51.313999496266284) -- (12.362997395756869, 51.313763242389527) -- (12.363995000516482, 51.311899434919425) -- (12.365051262574289, 51.311874076620185) -- (12.365607273029472, 51.310066772254423) -- (12.367888419286293, 51.305096585224398) -- (12.369674573591915, 51.304967455872415) -- (12.372537273655375, 51.302571373158798) -- (12.373086928998044, 51.301165079604004) -- (12.379545575204549, 51.299666971379082) -- (12.37872328759245, 51.298212804499961) -- (12.378921390603191, 51.29727894345843) -- (12.379763225230846, 51.295991952475191) -- (12.381119015331068, 51.294987706263804) -- (12.383924177183824, 51.294093603416748) -- (12.386512362349434, 51.293660424501084) -- (12.389192648396918, 51.290668357656699) -- (12.387995633329179, 51.289668118388406) -- (12.387473747627904, 51.290413888067413) -- (12.385350791642892, 51.290262653493834) -- (12.385057800671893, 51.289804036134811) -- (12.381065358211419, 51.288987702341863) -- (12.377334617220791, 51.287971108068191) -- (12.375926843008681, 51.291272742814236) -- (12.374910342900066, 51.292508760459086) -- (12.37300172153798, 51.292955671587016) -- (12.371290236206972, 51.293153338091798) -- (12.370511859464582, 51.291110323034957) -- (12.369462259533334, 51.290022259850353) -- (12.36454034937703, 51.289959109934649) -- (12.362420974717045, 51.291903825165875) -- (12.360083167437903, 51.292796369108409) -- (12.358273361117101, 51.292718268534998) -- (12.357712950055422, 51.291838759381911) -- (12.355791296729445, 51.293606601141796) -- (12.354171523467116, 51.294559231622756) -- (12.352414308467877, 51.295301425808766) -- (12.350630974187782, 51.295432018276152) -- (12.348374522962821, 51.295403685406008) -- (12.342410999386157, 51.294801344483417) -- (12.343446198950682, 51.297354061499775) -- (12.342943438819521, 51.30071634443685) -- (12.344034825445529, 51.304040806650157) -- (12.343555532700668, 51.306082737522757) -- (12.343278938466336, 51.306738587359575) -- (12.345592045420219, 51.308126661912382) -- (12.346437286017204, 51.308793267939492) -- (12.347353938594205, 51.309508555988543) -- (12.348880224960007, 51.311222686411064) -- (12.350055199360009, 51.312683311402537) -- (12.351202153992402, 51.314568733611019) -- (12.351935129380868, 51.316675396602967) -- (12.352188800182851, 51.318364204327793) -- (12.351940213359418, 51.318853258844456) -- (12.351641954959257, 51.31944097097989) -- (12.350388470599748, 51.320347897807686) -- (12.348172259629171, 51.321417335084192) -- (12.345990332673312, 51.322030266067316) -- (12.346490708957655, 51.324833275483222) -- (12.347353460398706, 51.327208587621207) -- (12.34836525703229, 51.329421467505341) -- (12.348559679448288, 51.332634488908042) -- (12.344419357022041, 51.333669329491201) -- (12.342416238561107, 51.334167976991068) -- (12.340703846780286, 51.334227278551197) -- (12.341664189855061, 51.335511184311514) -- (12.342043977133882, 51.336022046304308) -- (12.
342270644067744, 51.337729514149935) -- (12.342342068056602, 51.337755856622543) -- (12.346176128477437, 51.339026334031388) -- (12.346268259457363, 51.338851763161273) -- (12.347353684019337, 51.33710045182044) -- (12.348738598565443, 51.335337963082502) -- (12.350457680302574, 51.335741860486834);

        \draw[fill=NavyBlue, line width=0.01pt] (12.246522753244609, 51.384075217037221) -- (12.247393832331706, 51.384372062492751) -- (12.249163829478919, 51.384785924286867) -- (12.249415284430482, 51.384768019121104) -- (12.250445526525292, 51.385298359110479) -- (12.25147057030267, 51.384750516173575) -- (12.251535098380623, 51.384719085433908) -- (12.249654309392385, 51.382077249513983) -- (12.250816656546887, 51.381651442548353) -- (12.252197359372134, 51.381849688619759) -- (12.25271828098597, 51.381020470191942) -- (12.253766194393881, 51.380730110400258) -- (12.253888913368309, 51.377456252920126) -- (12.259271130637494, 51.375722038656932) -- (12.260885093595359, 51.375061607050284) -- (12.262352700674414, 51.375576271034383) -- (12.264159634897592, 51.376532200546407) -- (12.265236264932071, 51.376311137787653) -- (12.266335283287249, 51.378085757914015) -- (12.268697296117351, 51.379555188036996) -- (12.27108572302801, 51.378798948027466) -- (12.270887879624965, 51.377845254013053) -- (12.272487634123451, 51.377432333310985) -- (12.273498662599771, 51.375288410183721) -- (12.27380881647413, 51.373826273197267) -- (12.27313748501723, 51.371055063531912) -- (12.273277240542997, 51.370043659351246) -- (12.274438352045674, 51.36950772665967) -- (12.276010701822814, 51.36952599330121) -- (12.276695558642889, 51.370266910024725) -- (12.27920018445908, 51.370102219271558) -- (12.280096156493189, 51.369786586921151) -- (12.280868623471511, 51.369106962205244) -- (12.282152852820929, 51.369003327444418) -- (12.282909695584273, 51.369832290103048) -- (12.282583808915742, 51.370885012220768) -- (12.285480368612724, 51.370253988591884) -- (12.286542245745107, 51.371971310788872) -- (12.284743954871489, 51.372654544283556) -- (12.285617613345947, 51.373942359654343) -- (12.283289090612593, 51.374879011447234) -- (12.281813680648778, 51.373893277483816) -- (12.279653525205013, 51.375761525447771) -- (12.280236922803141, 51.376232879931557) -- (12.280870718316358, 51.376739875609651) -- (12.282016629140035, 51.377834282255868) -- (12.282428735362069, 51.378533158196767) -- (12.283810966985126, 51.377810793806972) -- (12.284551938088933, 51.378210890321448) -- (12.285447715690598, 51.377893797951408) -- (12.28726608794922, 51.37685398694606) -- (12.288633682489589, 51.376161590354052) -- (12.289334172401844, 51.377193200583164) -- (12.290954638092328, 51.376940291064393) -- (12.29130423239611, 51.378115618172401) -- (12.296281390888996, 51.377762572473443) -- (12.29969187538671, 51.376950178516097) -- (12.299987086197998, 51.375641304176433) -- (12.302524948741569, 51.375351667977668) -- (12.306842510490204, 51.37525786524693) -- (12.310553507338504, 51.374333876474012) -- (12.312428592008983, 51.373491247994451) -- (12.314735468685871, 51.374610329100456) -- (12.319036580904704, 51.371726907262691) -- (12.320019886486373, 51.370566273950537) -- (12.321066342504787, 51.370727793677368) -- (12.322451702406559, 51.371437542380377) -- (12.324263967567594, 51.371636206215392) -- (12.325831351407274, 51.37160326616663) -- (12.32882323882083, 51.370348891019866) -- (12.329632057738772, 51.368982837480445) -- (12.333575239733351, 51.368290860162546) -- (12.335622230574739, 51.367495138798176) -- (12.336847136310912, 51.366712047608736) -- (12.336000701166544, 51.366001475159202) -- (12.338461789447825, 51.360832288791457) -- (12.341129193123255, 51.36186309942056) -- (12.34195650811335, 51.360224780975727) -- (12.344341513596744, 51.359113696037994) -- (12.345136446831756, 51.359787829310676) -- (12.346735493632552, 51.358573530505737) -- (12.347600792737344, 51.357568696531743) -- (12.35015624313213, 51.35671011433756) -- (12.352275044928382, 51.35711708078712) -- (12.355485947000082, 51.354999292113362) -- (12.35832209243811, 51.35662900342497) -- (12.359445439245739, 51.35618882106079) -- (12.362241222487967, 51.355492279136996) -- (12.364495716724051, 51.354852519015076) -- (12.365319549672318, 51.354396114121855) -- (12.367677737660886, 51.353165934808551) -- (12.369392736796422, 51.352926756915281) -- (12.36984169500715, 51.
351503681367561) -- (12.368564823621767, 51.351003258007282) -- (12.369223947483182, 51.350311194184769) -- (12.370352998643272, 51.349606977605731) -- (12.369851845600937, 51.347190620284501) -- (12.368009434018132, 51.347609557387749) -- (12.365737688287201, 51.34748229015446) -- (12.361164044798107, 51.346848302016092) -- (12.355994746647676, 51.349020015171874) -- (12.357658987482621, 51.350182885602777) -- (12.354004340495889, 51.350090335631627) -- (12.353282344220446, 51.351255093780189) -- (12.348831933359381, 51.350354508052533) -- (12.346311976552157, 51.353051007964702) -- (12.344517853737383, 51.35199511390072) -- (12.343404339602083, 51.351765068543145) -- (12.341345629308595, 51.352508274911536) -- (12.341076677906312, 51.353058477379868) -- (12.340291658571687, 51.354798218401335) -- (12.33977537650968, 51.355934269469216) -- (12.338198657694754, 51.356107697058178) -- (12.336758592041136, 51.358323816021944) -- (12.335715219317326, 51.358146109064059) -- (12.33651441206519, 51.356277026331561) -- (12.33353315389297, 51.357313588586649) -- (12.333122373602066, 51.356436679959451) -- (12.33723804256798, 51.355135645118878) -- (12.338312617804823, 51.354295765937707) -- (12.342005117713381, 51.347700461785969) -- (12.339771743403764, 51.347659560881731) -- (12.339021377452122, 51.34736935182292) -- (12.338145351555761, 51.343882052985677) -- (12.334104078488531, 51.343939429636244) -- (12.334208436228492, 51.345165267445587) -- (12.327811416678017, 51.345031704900158) -- (12.326931383170262, 51.346895114437444) -- (12.327814971934759, 51.34795678434886) -- (12.326212706909454, 51.347923861387216) -- (12.323977965678798, 51.351117154271265) -- (12.32599683224743, 51.351411765920325) -- (12.326634728325402, 51.351743415338447) -- (12.326603265275594, 51.352876257338892) -- (12.325826521488141, 51.353651145174013) -- (12.326372348581211, 51.356043841924567) -- (12.323706028339783, 51.356692835029129) -- (12.318572767864996, 51.356935508239829) -- (12.318365016232223, 51.356940722877987) -- (12.317838545365285, 51.357879263415711) -- (12.314636157536155, 51.358260729602989) -- (12.313942772650453, 51.356979426343052) -- (12.311870153345206, 51.356847496509239) -- (12.311741000486265, 51.356825949081063) -- (12.311414890693008, 51.357763935971192) -- (12.309736525560789, 51.357667320822657) -- (12.308901940797451, 51.358782565020512) -- (12.30860839859281, 51.360683275497514) -- (12.302067429369602, 51.36072216101087) -- (12.300204670240651, 51.361253211641476) -- (12.29897523894412, 51.362295542285459) -- (12.294893281479373, 51.362134065384943) -- (12.292976647459856, 51.361755851946569) -- (12.29247888874972, 51.364507023295516) -- (12.291283218060554, 51.365002892837012) -- (12.289297321331267, 51.365190711121755) -- (12.287333893244178, 51.365643213027894) -- (12.286573135568434, 51.367297759182676) -- (12.283978188948126, 51.36787279504874) -- (12.283318190508282, 51.366117002088174) -- (12.281746586784671, 51.36608473512895) -- (12.281477511731753, 51.365176116131728) -- (12.28039098594331, 51.364428798439732) -- (12.279227603716308, 51.364217859264741) -- (12.277642273343293, 51.364427907161243) -- (12.274905197989451, 51.363822518882955) -- (12.275001952726143, 51.364743700392999) -- (12.275091059443319, 51.365536595874865) -- (12.271440587291558, 51.366593565979514) -- (12.270422662707242, 51.367062062308619) -- (12.270489858935029, 51.367762191973824) -- (12.270922902640701, 51.368356553461652) -- (12.272059656836335, 51.368951737505981) -- (12.26937012395509, 51.369848641217544) -- (12.267475676358472, 51.370271444919716) -- (12.269696899969295, 51.366928260472775) -- (12.267004223809117, 51.369020067211252) -- (12.265774510392861, 51.370565570680945) -- (12.263152048553426, 51.370404783146533) -- (12.26302620893426, 51.369384238102391) -- (12.263816131530968, 51.367311711123655) -- (12.260364409727808, 51.366292946190796) -- (12.261799952116782, 51.363857880878328) -- (12.259521110551093, 51.363386894067759) -- (12.257257735681563, 51.363300098103672) -- (12.248183285952161, 51.36373488862143) -- (
12.247690183630381, 51.364600939469582) -- (12.246814861646913, 51.36507150134041) -- (12.245807838577376, 51.364836957468761) -- (12.2438414232022, 51.365773327486622) -- (12.242426683808604, 51.366180294147348) -- (12.233934642339268, 51.367788272869532) -- (12.234484927288017, 51.368934273987563) -- (12.233465776278955, 51.369447884088252) -- (12.230922483090223, 51.367880553669963) -- (12.22953548944434, 51.370177079255811) -- (12.230099422680977, 51.37105808334271) -- (12.229842750557417, 51.371660584527703) -- (12.231052857454873, 51.372262195645582) -- (12.230884004297474, 51.373259069857738) -- (12.229302474513062, 51.373332142595928) -- (12.230802855177876, 51.378160297483859) -- (12.229696317568656, 51.378334826817017) -- (12.228186214683037, 51.376558596519125) -- (12.227115492357047, 51.376714832637248) -- (12.225691352683789, 51.376605823482862) -- (12.225578943352197, 51.375464919097098) -- (12.227448529904082, 51.375489440295496) -- (12.227565725248423, 51.374193644304285) -- (12.226155072999806, 51.373837495877531) -- (12.225410889648259, 51.372882074620115) -- (12.225348309546142, 51.371494457069495) -- (12.223764747762653, 51.371017158467502) -- (12.221369018889131, 51.370632276045576) -- (12.218067307565756, 51.369682497044671) -- (12.217878274688685, 51.369628115558939) -- (12.216799331652942, 51.371601486104396) -- (12.214367642473924, 51.371217582450072) -- (12.213139333919457, 51.371575158867969) -- (12.214399114541607, 51.371993547359359) -- (12.214432130483928, 51.373183108670702) -- (12.214825602475132, 51.374585306376218) -- (12.21352525238807, 51.374895045775048) -- (12.212472506513214, 51.374112315159834) -- (12.210920030681038, 51.374300700937731) -- (12.208855089210765, 51.373954015339706) -- (12.207395123716301, 51.373868363441403) -- (12.206677229822247, 51.37419202589119) -- (12.205979325200865, 51.376071468299806) -- (12.206664951934053, 51.37864800143619) -- (12.209164683146307, 51.377915459998064) -- (12.21003067526938, 51.378693549586004) -- (12.208660226043499, 51.379480155630617) -- (12.207316992164763, 51.379880122418115) -- (12.205355005800122, 51.380217142385057) -- (12.203485523737612, 51.380306993062355) -- (12.196573449841249, 51.380038137495852) -- (12.189680711363202, 51.378930685215899) -- (12.188973074299239, 51.380239222226272) -- (12.190193915401418, 51.381888291124341) -- (12.189592515705305, 51.382675243390416) -- (12.189272732398623, 51.383785672531978) -- (12.190384481971387, 51.382922087469261) -- (12.192672664883768, 51.383231299206834) -- (12.193606339068129, 51.38480885197319) -- (12.190887846164504, 51.384904201159678) -- (12.19158950649085, 51.385969259563645) -- (12.192863426489581, 51.386162137524138) -- (12.192311861490763, 51.386764643605304) -- (12.192998436574038, 51.387582547443245) -- (12.194429852690803, 51.387550753040934) -- (12.19494056629668, 51.387541663233883) -- (12.195577605755629, 51.388517024151348) -- (12.196484968817158, 51.388737181490043) -- (12.198515207773884, 51.389092182978622) -- (12.198920615257807, 51.389676584445702) -- (12.200068012383136, 51.388835931904715) -- (12.200945775506133, 51.388858438214989) -- (12.201344427512213, 51.389532811064811) -- (12.202510721082682, 51.389721734104974) -- (12.203342958427397, 51.3892183046855) -- (12.202590734446687, 51.388305569265626) -- (12.204262696533192, 51.387366094075006) -- (12.205623761751248, 51.386075892373796) -- (12.206976600504332, 51.386202012809129) -- (12.208198227474099, 51.385936980417334) -- (12.21058531553262, 51.384103026983574) -- (12.211701154056778, 51.384337165409981) -- (12.211780736382146, 51.386050344368329) -- (12.215134025612725, 51.386388744007398) -- (12.216442602397949, 51.384734555916474) -- (12.218911314728679, 51.385140302577859) -- (12.220254642979143, 51.384808187210901) -- (12.220820148197198, 51.382650182034233) -- (12.223481853697647, 51.382858396841108) -- (12.223767471441324, 51.382408927078131) -- (12.225067750141895, 51.382121896291295) -- (12.225287937646275, 51.381510548831692) -- (12.226181343522086, 51.381902228068697) -- (12.227799005632077,
 51.381849234565024) -- (12.229384499486233, 51.381045459890373) -- (12.229813750119721, 51.380555707046931) -- (12.229874475161107, 51.379643142340882) -- (12.230117195541764, 51.379238710595864) -- (12.231849210872458, 51.379055597880139) -- (12.232584834407923, 51.379577555532812) -- (12.23302526762907, 51.378332801639495) -- (12.235542600598729, 51.378559573089518) -- (12.237263495397762, 51.379176422542287) -- (12.238939239387522, 51.379335028002657) -- (12.24180027646884, 51.379315207596647) -- (12.241696926658449, 51.380343292881619) -- (12.243075743052062, 51.380034492603926) -- (12.243473315602097, 51.378881899864723) -- (12.246480180663648, 51.379418772581886) -- (12.246393296626321, 51.380921424997098) -- (12.248846816864127, 51.38151110727091) -- (12.247146944640511, 51.382287280461576) -- (12.246023435322801, 51.38335612324277) -- (12.246097988045527, 51.383926783297113) -- (12.246522753244609, 51.384075217037221);
        
        \draw[step=0.5mm,Gray!20!black,line width=0.01pt] (12.18897307429924, 51.2385) grid (12.5392, 51.4436);
        
        \draw[fill=YellowGreen!80!white,YellowGreen!50!white,opacity=0.5,line width=0.01pt] (12.3, 51.25) -- (12.3, 51.35) -- (12.4, 51.35) -- (12.4, 51.25) -- (12.3, 51.25);
        
        \draw[YellowGreen!80!white, fill=YellowGreen!80!white, pattern=south west lines,pattern color=YellowGreen!20!black,opacity=0.5,line width=0.01pt] (12.3, 51.25) -- (12.3, 51.35) -- (12.4, 51.35) -- (12.4, 51.25) -- (12.3, 51.25);
        
        
        \draw[red,ultra thin] (12.18897307429924, 51.34388205298568) -- (12.18897307429924, 51.389721734104974) -- (12.370352998643272, 51.389721734104974) -- (12.370352998643272, 51.34388205298568) -- (12.18897307429924, 51.34388205298568);
        
        \draw[red,ultra thin] (12.340703846780286, 51.28797110806819) -- (12.340703846780286, 51.33902633403139) -- (12.389192648396918, 51.33902633403139) -- (12.389192648396918, 51.28797110806819) -- (12.340703846780286, 51.28797110806819);

        \draw[red,ultra thin] (12.2433, 51.2385) --(12.2433, 51.4436) -- (12.5392, 51.4436) -- (12.5392, 51.2385) -- (12.2433, 51.2385);
        
        \end{tikzpicture}
    }
}~\subfigure[Optimized sparse space tiling]{\label{fig:leipzig4}
    \resizebox{0.45\linewidth}{!}{
        \begin{tikzpicture}[scale=3]

        \draw[fill=Gray!20!white, line width=0.01pt] (12.24334845000007, 51.348402) -- (12.260675450000065, 51.42005750000007) -- (12.45318745000003, 51.44363850000005) -- (12.498737950000077, 51.36050450000005) -- (12.5391649500001, 51.3063565) -- (12.253649450000097, 51.23853550000001) -- (12.24334845000007, 51.348402);

        \draw[fill=Green, line width=0.01pt]
        (12.350457680302574, 51.335741860486834) -- (12.352140928324896, 51.333637635911934) -- (12.352439235513609, 51.333260718584498) -- (12.35313583488478, 51.332573109900757) -- (12.353698084666826, 51.332288104893777) -- (12.354745816406369, 51.332356893302503) -- (12.357841093431484, 51.333607855296478) -- (12.362974379474052, 51.335692068272955) -- (12.363614519530703, 51.335366147209989) -- (12.365929215946238, 51.335501366978214) -- (12.366311547463905, 51.334507614038031) -- (12.364308429206481, 51.333575625095392) -- (12.359845102391926, 51.331504636298114) -- (12.363055148685051, 51.33018700453627) -- (12.364807826900106, 51.325985154216667) -- (12.361284026250093, 51.327881992318808) -- (12.359637477398138, 51.327576574910601) -- (12.358879859135268, 51.326425098612404) -- (12.357259861513205, 51.325364351195581) -- (12.357642647938114, 51.324289466492395) -- (12.357946704334694, 51.322987080911844) -- (12.358780590768848, 51.318658908839012) -- (12.358634489819959, 51.316495595865717) -- (12.360475078143981, 51.316603469281716) -- (12.3617195173801, 51.319237258521404) -- (12.36541261590396, 51.318790499281072) -- (12.365630387854567, 51.313999496266284) -- (12.362997395756869, 51.313763242389527) -- (12.363995000516482, 51.311899434919425) -- (12.365051262574289, 51.311874076620185) -- (12.365607273029472, 51.310066772254423) -- (12.367888419286293, 51.305096585224398) -- (12.369674573591915, 51.304967455872415) -- (12.372537273655375, 51.302571373158798) -- (12.373086928998044, 51.301165079604004) -- (12.379545575204549, 51.299666971379082) -- (12.37872328759245, 51.298212804499961) -- (12.378921390603191, 51.29727894345843) -- (12.379763225230846, 51.295991952475191) -- (12.381119015331068, 51.294987706263804) -- (12.383924177183824, 51.294093603416748) -- (12.386512362349434, 51.293660424501084) -- (12.389192648396918, 51.290668357656699) -- (12.387995633329179, 51.289668118388406) -- (12.387473747627904, 51.290413888067413) -- (12.385350791642892, 51.290262653493834) -- (12.385057800671893, 51.289804036134811) -- (12.381065358211419, 51.288987702341863) -- (12.377334617220791, 51.287971108068191) -- (12.375926843008681, 51.291272742814236) -- (12.374910342900066, 51.292508760459086) -- (12.37300172153798, 51.292955671587016) -- (12.371290236206972, 51.293153338091798) -- (12.370511859464582, 51.291110323034957) -- (12.369462259533334, 51.290022259850353) -- (12.36454034937703, 51.289959109934649) -- (12.362420974717045, 51.291903825165875) -- (12.360083167437903, 51.292796369108409) -- (12.358273361117101, 51.292718268534998) -- (12.357712950055422, 51.291838759381911) -- (12.355791296729445, 51.293606601141796) -- (12.354171523467116, 51.294559231622756) -- (12.352414308467877, 51.295301425808766) -- (12.350630974187782, 51.295432018276152) -- (12.348374522962821, 51.295403685406008) -- (12.342410999386157, 51.294801344483417) -- (12.343446198950682, 51.297354061499775) -- (12.342943438819521, 51.30071634443685) -- (12.344034825445529, 51.304040806650157) -- (12.343555532700668, 51.306082737522757) -- (12.343278938466336, 51.306738587359575) -- (12.345592045420219, 51.308126661912382) -- (12.346437286017204, 51.308793267939492) -- (12.347353938594205, 51.309508555988543) -- (12.348880224960007, 51.311222686411064) -- (12.350055199360009, 51.312683311402537) -- (12.351202153992402, 51.314568733611019) -- (12.351935129380868, 51.316675396602967) -- (12.352188800182851, 51.318364204327793) -- (12.351940213359418, 51.318853258844456) -- (12.351641954959257, 51.31944097097989) -- (12.350388470599748, 51.320347897807686) -- (12.348172259629171, 51.321417335084192) -- (12.345990332673312, 51.322030266067316) -- (12.346490708957655, 51.324833275483222) -- (12.347353460398706, 51.327208587621207) -- (12.34836525703229, 51.329421467505341) -- (12.348559679448288, 51.332634488908042) -- (12.344419357022041, 51.333669329491201) -- (12.342416238561107, 51.334167976991068) -- (12.340703846780286, 51.334227278551197) -- (12.341664189855061, 51.335511184311514) -- (12.342043977133882, 51.336022046304308) -- (12.
342270644067744, 51.337729514149935) -- (12.342342068056602, 51.337755856622543) -- (12.346176128477437, 51.339026334031388) -- (12.346268259457363, 51.338851763161273) -- (12.347353684019337, 51.33710045182044) -- (12.348738598565443, 51.335337963082502) -- (12.350457680302574, 51.335741860486834);

        \draw[fill=NavyBlue, line width=0.01pt] (12.246522753244609, 51.384075217037221) -- (12.247393832331706, 51.384372062492751) -- (12.249163829478919, 51.384785924286867) -- (12.249415284430482, 51.384768019121104) -- (12.250445526525292, 51.385298359110479) -- (12.25147057030267, 51.384750516173575) -- (12.251535098380623, 51.384719085433908) -- (12.249654309392385, 51.382077249513983) -- (12.250816656546887, 51.381651442548353) -- (12.252197359372134, 51.381849688619759) -- (12.25271828098597, 51.381020470191942) -- (12.253766194393881, 51.380730110400258) -- (12.253888913368309, 51.377456252920126) -- (12.259271130637494, 51.375722038656932) -- (12.260885093595359, 51.375061607050284) -- (12.262352700674414, 51.375576271034383) -- (12.264159634897592, 51.376532200546407) -- (12.265236264932071, 51.376311137787653) -- (12.266335283287249, 51.378085757914015) -- (12.268697296117351, 51.379555188036996) -- (12.27108572302801, 51.378798948027466) -- (12.270887879624965, 51.377845254013053) -- (12.272487634123451, 51.377432333310985) -- (12.273498662599771, 51.375288410183721) -- (12.27380881647413, 51.373826273197267) -- (12.27313748501723, 51.371055063531912) -- (12.273277240542997, 51.370043659351246) -- (12.274438352045674, 51.36950772665967) -- (12.276010701822814, 51.36952599330121) -- (12.276695558642889, 51.370266910024725) -- (12.27920018445908, 51.370102219271558) -- (12.280096156493189, 51.369786586921151) -- (12.280868623471511, 51.369106962205244) -- (12.282152852820929, 51.369003327444418) -- (12.282909695584273, 51.369832290103048) -- (12.282583808915742, 51.370885012220768) -- (12.285480368612724, 51.370253988591884) -- (12.286542245745107, 51.371971310788872) -- (12.284743954871489, 51.372654544283556) -- (12.285617613345947, 51.373942359654343) -- (12.283289090612593, 51.374879011447234) -- (12.281813680648778, 51.373893277483816) -- (12.279653525205013, 51.375761525447771) -- (12.280236922803141, 51.376232879931557) -- (12.280870718316358, 51.376739875609651) -- (12.282016629140035, 51.377834282255868) -- (12.282428735362069, 51.378533158196767) -- (12.283810966985126, 51.377810793806972) -- (12.284551938088933, 51.378210890321448) -- (12.285447715690598, 51.377893797951408) -- (12.28726608794922, 51.37685398694606) -- (12.288633682489589, 51.376161590354052) -- (12.289334172401844, 51.377193200583164) -- (12.290954638092328, 51.376940291064393) -- (12.29130423239611, 51.378115618172401) -- (12.296281390888996, 51.377762572473443) -- (12.29969187538671, 51.376950178516097) -- (12.299987086197998, 51.375641304176433) -- (12.302524948741569, 51.375351667977668) -- (12.306842510490204, 51.37525786524693) -- (12.310553507338504, 51.374333876474012) -- (12.312428592008983, 51.373491247994451) -- (12.314735468685871, 51.374610329100456) -- (12.319036580904704, 51.371726907262691) -- (12.320019886486373, 51.370566273950537) -- (12.321066342504787, 51.370727793677368) -- (12.322451702406559, 51.371437542380377) -- (12.324263967567594, 51.371636206215392) -- (12.325831351407274, 51.37160326616663) -- (12.32882323882083, 51.370348891019866) -- (12.329632057738772, 51.368982837480445) -- (12.333575239733351, 51.368290860162546) -- (12.335622230574739, 51.367495138798176) -- (12.336847136310912, 51.366712047608736) -- (12.336000701166544, 51.366001475159202) -- (12.338461789447825, 51.360832288791457) -- (12.341129193123255, 51.36186309942056) -- (12.34195650811335, 51.360224780975727) -- (12.344341513596744, 51.359113696037994) -- (12.345136446831756, 51.359787829310676) -- (12.346735493632552, 51.358573530505737) -- (12.347600792737344, 51.357568696531743) -- (12.35015624313213, 51.35671011433756) -- (12.352275044928382, 51.35711708078712) -- (12.355485947000082, 51.354999292113362) -- (12.35832209243811, 51.35662900342497) -- (12.359445439245739, 51.35618882106079) -- (12.362241222487967, 51.355492279136996) -- (12.364495716724051, 51.354852519015076) -- (12.365319549672318, 51.354396114121855) -- (12.367677737660886, 51.353165934808551) -- (12.369392736796422, 51.352926756915281) -- (12.36984169500715, 51.
351503681367561) -- (12.368564823621767, 51.351003258007282) -- (12.369223947483182, 51.350311194184769) -- (12.370352998643272, 51.349606977605731) -- (12.369851845600937, 51.347190620284501) -- (12.368009434018132, 51.347609557387749) -- (12.365737688287201, 51.34748229015446) -- (12.361164044798107, 51.346848302016092) -- (12.355994746647676, 51.349020015171874) -- (12.357658987482621, 51.350182885602777) -- (12.354004340495889, 51.350090335631627) -- (12.353282344220446, 51.351255093780189) -- (12.348831933359381, 51.350354508052533) -- (12.346311976552157, 51.353051007964702) -- (12.344517853737383, 51.35199511390072) -- (12.343404339602083, 51.351765068543145) -- (12.341345629308595, 51.352508274911536) -- (12.341076677906312, 51.353058477379868) -- (12.340291658571687, 51.354798218401335) -- (12.33977537650968, 51.355934269469216) -- (12.338198657694754, 51.356107697058178) -- (12.336758592041136, 51.358323816021944) -- (12.335715219317326, 51.358146109064059) -- (12.33651441206519, 51.356277026331561) -- (12.33353315389297, 51.357313588586649) -- (12.333122373602066, 51.356436679959451) -- (12.33723804256798, 51.355135645118878) -- (12.338312617804823, 51.354295765937707) -- (12.342005117713381, 51.347700461785969) -- (12.339771743403764, 51.347659560881731) -- (12.339021377452122, 51.34736935182292) -- (12.338145351555761, 51.343882052985677) -- (12.334104078488531, 51.343939429636244) -- (12.334208436228492, 51.345165267445587) -- (12.327811416678017, 51.345031704900158) -- (12.326931383170262, 51.346895114437444) -- (12.327814971934759, 51.34795678434886) -- (12.326212706909454, 51.347923861387216) -- (12.323977965678798, 51.351117154271265) -- (12.32599683224743, 51.351411765920325) -- (12.326634728325402, 51.351743415338447) -- (12.326603265275594, 51.352876257338892) -- (12.325826521488141, 51.353651145174013) -- (12.326372348581211, 51.356043841924567) -- (12.323706028339783, 51.356692835029129) -- (12.318572767864996, 51.356935508239829) -- (12.318365016232223, 51.356940722877987) -- (12.317838545365285, 51.357879263415711) -- (12.314636157536155, 51.358260729602989) -- (12.313942772650453, 51.356979426343052) -- (12.311870153345206, 51.356847496509239) -- (12.311741000486265, 51.356825949081063) -- (12.311414890693008, 51.357763935971192) -- (12.309736525560789, 51.357667320822657) -- (12.308901940797451, 51.358782565020512) -- (12.30860839859281, 51.360683275497514) -- (12.302067429369602, 51.36072216101087) -- (12.300204670240651, 51.361253211641476) -- (12.29897523894412, 51.362295542285459) -- (12.294893281479373, 51.362134065384943) -- (12.292976647459856, 51.361755851946569) -- (12.29247888874972, 51.364507023295516) -- (12.291283218060554, 51.365002892837012) -- (12.289297321331267, 51.365190711121755) -- (12.287333893244178, 51.365643213027894) -- (12.286573135568434, 51.367297759182676) -- (12.283978188948126, 51.36787279504874) -- (12.283318190508282, 51.366117002088174) -- (12.281746586784671, 51.36608473512895) -- (12.281477511731753, 51.365176116131728) -- (12.28039098594331, 51.364428798439732) -- (12.279227603716308, 51.364217859264741) -- (12.277642273343293, 51.364427907161243) -- (12.274905197989451, 51.363822518882955) -- (12.275001952726143, 51.364743700392999) -- (12.275091059443319, 51.365536595874865) -- (12.271440587291558, 51.366593565979514) -- (12.270422662707242, 51.367062062308619) -- (12.270489858935029, 51.367762191973824) -- (12.270922902640701, 51.368356553461652) -- (12.272059656836335, 51.368951737505981) -- (12.26937012395509, 51.369848641217544) -- (12.267475676358472, 51.370271444919716) -- (12.269696899969295, 51.366928260472775) -- (12.267004223809117, 51.369020067211252) -- (12.265774510392861, 51.370565570680945) -- (12.263152048553426, 51.370404783146533) -- (12.26302620893426, 51.369384238102391) -- (12.263816131530968, 51.367311711123655) -- (12.260364409727808, 51.366292946190796) -- (12.261799952116782, 51.363857880878328) -- (12.259521110551093, 51.363386894067759) -- (12.257257735681563, 51.363300098103672) -- (12.248183285952161, 51.36373488862143) -- (
12.247690183630381, 51.364600939469582) -- (12.246814861646913, 51.36507150134041) -- (12.245807838577376, 51.364836957468761) -- (12.2438414232022, 51.365773327486622) -- (12.242426683808604, 51.366180294147348) -- (12.233934642339268, 51.367788272869532) -- (12.234484927288017, 51.368934273987563) -- (12.233465776278955, 51.369447884088252) -- (12.230922483090223, 51.367880553669963) -- (12.22953548944434, 51.370177079255811) -- (12.230099422680977, 51.37105808334271) -- (12.229842750557417, 51.371660584527703) -- (12.231052857454873, 51.372262195645582) -- (12.230884004297474, 51.373259069857738) -- (12.229302474513062, 51.373332142595928) -- (12.230802855177876, 51.378160297483859) -- (12.229696317568656, 51.378334826817017) -- (12.228186214683037, 51.376558596519125) -- (12.227115492357047, 51.376714832637248) -- (12.225691352683789, 51.376605823482862) -- (12.225578943352197, 51.375464919097098) -- (12.227448529904082, 51.375489440295496) -- (12.227565725248423, 51.374193644304285) -- (12.226155072999806, 51.373837495877531) -- (12.225410889648259, 51.372882074620115) -- (12.225348309546142, 51.371494457069495) -- (12.223764747762653, 51.371017158467502) -- (12.221369018889131, 51.370632276045576) -- (12.218067307565756, 51.369682497044671) -- (12.217878274688685, 51.369628115558939) -- (12.216799331652942, 51.371601486104396) -- (12.214367642473924, 51.371217582450072) -- (12.213139333919457, 51.371575158867969) -- (12.214399114541607, 51.371993547359359) -- (12.214432130483928, 51.373183108670702) -- (12.214825602475132, 51.374585306376218) -- (12.21352525238807, 51.374895045775048) -- (12.212472506513214, 51.374112315159834) -- (12.210920030681038, 51.374300700937731) -- (12.208855089210765, 51.373954015339706) -- (12.207395123716301, 51.373868363441403) -- (12.206677229822247, 51.37419202589119) -- (12.205979325200865, 51.376071468299806) -- (12.206664951934053, 51.37864800143619) -- (12.209164683146307, 51.377915459998064) -- (12.21003067526938, 51.378693549586004) -- (12.208660226043499, 51.379480155630617) -- (12.207316992164763, 51.379880122418115) -- (12.205355005800122, 51.380217142385057) -- (12.203485523737612, 51.380306993062355) -- (12.196573449841249, 51.380038137495852) -- (12.189680711363202, 51.378930685215899) -- (12.188973074299239, 51.380239222226272) -- (12.190193915401418, 51.381888291124341) -- (12.189592515705305, 51.382675243390416) -- (12.189272732398623, 51.383785672531978) -- (12.190384481971387, 51.382922087469261) -- (12.192672664883768, 51.383231299206834) -- (12.193606339068129, 51.38480885197319) -- (12.190887846164504, 51.384904201159678) -- (12.19158950649085, 51.385969259563645) -- (12.192863426489581, 51.386162137524138) -- (12.192311861490763, 51.386764643605304) -- (12.192998436574038, 51.387582547443245) -- (12.194429852690803, 51.387550753040934) -- (12.19494056629668, 51.387541663233883) -- (12.195577605755629, 51.388517024151348) -- (12.196484968817158, 51.388737181490043) -- (12.198515207773884, 51.389092182978622) -- (12.198920615257807, 51.389676584445702) -- (12.200068012383136, 51.388835931904715) -- (12.200945775506133, 51.388858438214989) -- (12.201344427512213, 51.389532811064811) -- (12.202510721082682, 51.389721734104974) -- (12.203342958427397, 51.3892183046855) -- (12.202590734446687, 51.388305569265626) -- (12.204262696533192, 51.387366094075006) -- (12.205623761751248, 51.386075892373796) -- (12.206976600504332, 51.386202012809129) -- (12.208198227474099, 51.385936980417334) -- (12.21058531553262, 51.384103026983574) -- (12.211701154056778, 51.384337165409981) -- (12.211780736382146, 51.386050344368329) -- (12.215134025612725, 51.386388744007398) -- (12.216442602397949, 51.384734555916474) -- (12.218911314728679, 51.385140302577859) -- (12.220254642979143, 51.384808187210901) -- (12.220820148197198, 51.382650182034233) -- (12.223481853697647, 51.382858396841108) -- (12.223767471441324, 51.382408927078131) -- (12.225067750141895, 51.382121896291295) -- (12.225287937646275, 51.381510548831692) -- (12.226181343522086, 51.381902228068697) -- (12.227799005632077,
 51.381849234565024) -- (12.229384499486233, 51.381045459890373) -- (12.229813750119721, 51.380555707046931) -- (12.229874475161107, 51.379643142340882) -- (12.230117195541764, 51.379238710595864) -- (12.231849210872458, 51.379055597880139) -- (12.232584834407923, 51.379577555532812) -- (12.23302526762907, 51.378332801639495) -- (12.235542600598729, 51.378559573089518) -- (12.237263495397762, 51.379176422542287) -- (12.238939239387522, 51.379335028002657) -- (12.24180027646884, 51.379315207596647) -- (12.241696926658449, 51.380343292881619) -- (12.243075743052062, 51.380034492603926) -- (12.243473315602097, 51.378881899864723) -- (12.246480180663648, 51.379418772581886) -- (12.246393296626321, 51.380921424997098) -- (12.248846816864127, 51.38151110727091) -- (12.247146944640511, 51.382287280461576) -- (12.246023435322801, 51.38335612324277) -- (12.246097988045527, 51.383926783297113) -- (12.246522753244609, 51.384075217037221);
        
        \draw[step=0.5mm,Gray!20!black,line width=0.01pt] (12.18897307429924, 51.2385) grid (12.5392, 51.4436);

        \draw[Goldenrod, fill=Goldenrod, opacity=0.5,line width=0.01pt] (12.3, 51.30) -- (12.2, 51.30) -- (12.2, 51.4) -- (12.4, 51.40) -- (12.4, 51.30) -- (12.4, 51.25) -- (12.3, 51.25) -- (12.3, 51.30);
        
        \draw[Goldenrod, fill=Goldenrod, pattern=south west lines,pattern color=Goldenrod!20!black,opacity=0.5,line width=0.01pt] (12.3, 51.30) -- (12.2, 51.30) -- (12.2, 51.4) -- (12.4, 51.40) -- (12.4, 51.30) -- (12.4, 51.25) -- (12.3, 51.25) -- (12.3, 51.30);
        
        \draw[red,ultra thin] (12.18897307429924, 51.34388205298568) -- (12.18897307429924, 51.389721734104974) -- (12.370352998643272, 51.389721734104974) -- (12.370352998643272, 51.34388205298568) -- (12.18897307429924, 51.34388205298568);
        
        \draw[red,ultra thin] (12.340703846780286, 51.28797110806819) -- (12.340703846780286, 51.33902633403139) -- (12.389192648396918, 51.33902633403139) -- (12.389192648396918, 51.28797110806819) -- (12.340703846780286, 51.28797110806819);
        
        \draw[red,ultra thin] (12.2433, 51.2385) --(12.2433, 51.4436) -- (12.5392, 51.4436) -- (12.5392, 51.2385) -- (12.2433, 51.2385);
        
        \end{tikzpicture}
    }
}

\caption{City of Leipzig from NUTS (in gray) together with topologically related geometries from CLC (in green and blue). See \autoref{sec:evaluation} for description of NUTS and CLC.}
\label{fig:leipzig}
\end{figure}

\subsubsection{DE-9IM}

The Dimensionally Extended nine-Intersection Model (DE-9IM)~\cite{clementini1994modelling} is a standard used to describe the topological relations between two geometries in two-dimensional space.
The spatial relations expressed by the model are topological and are invariant to rotation, translation and scaling transformations~\cite{egenhofer1991point}.
The basic idea behind the DE-9IM model is to construct the $3 \times 3$ intersection matrix:
\begin{equation}\setlength\arraycolsep{8pt}
\label{mat:de9im}
    DE9IM(a,b)
    \begin{bmatrix}
        dim(I(g_1) \cap I(g_2)) & dim(I(g_1) \cap B(g_2)) & dim(I(g_1) \cap E(g_2)) \\
        dim(B(g_1) \cap I(g_2)) & dim(B(g_1) \cap B(g_2)) & dim(B(g_1) \cap E(g_2)) \\
        dim(E(g_1) \cap I(g_2)) & dim(E(g_1) \cap B(g_2)) & dim(E(g_1) \cap E(g_2)) \\
    \end{bmatrix}
\end{equation}
where $dim$ is the maximum number of dimensions of the intersection $\cap$ of the interior ($I$), boundary ($B$), or exterior ($E$) of the two geometries $g_1$ and $g_2$.
The domain of $dim$ is $\{-1, 0, 1 , 2\}$, where $-1$ indicates no intersection, $0$ stands for an intersection which results into a set of one or more points, $1$ indicates an  intersection made up of lines and $2$ standard for an intersection which results in an area.
A simplified binary version of $dim(x)$ with the binary domain $\{$\emph{true}, \emph{false}$\}$ is obtained using the boolean function $\beta(dim(I(g)) =$ \emph{false} iff $dim(I(g)) = -1$ and \emph{true} otherwise.

The major insight behind \radon is that \emph{one condition must hold for any of the entries of the DE-9IM matrix to be \emph{true}: There must be at least one point in space that is common to the shapes of the polygons}. Here, sharing common points includes the intersection of the lines connecting the points which make up the polygon. Note that the only spatial relation for which all entries are 0 is the \texttt{disjoint} relation, which \radon can easily compute by computing the inverse of the \texttt{intersects} relation. 
Hence, by accelerating the computation of whether two geometries share at least one point, we can accelerate the computation of any of the DE-9IM entries. Therewith, we can also accelerate the computation of any topological relation, as they can all be derived from the DE-9IM entries. We implement this insight by using an improved indexing approach based on minimum bounding boxes and space tiling.


The \emph{minimum bounding box (MBB)} of a geometry $g$ in $n$ dimensions ~\cite{o1985finding} (also called its \emph{envelope}) is the rectangular box with the smallest measure (area, volume, or hypervolume in higher dimensions) within which all points  of $g$ lie. 
Let $\kappa_i(p)$ denote the $i^\text{th}$ dimension coordinate of a point $p$. To obtain the MBB of a geometry $g$, we have to find the lowest point coordinate $c_i^{\bot}=\min_{p\in g}\{\kappa_i(p)\}$ and the highest point coordinate $c_i^{\top}=\max_{p\in g}\{\kappa_i(p)\}$ in each dimension $i \in \{0,\dots,n\}$. Then, the $2^n$ vertices of the MBB in $n$ dimensions are all the vectors $\left(c_0^{(\cdot)}, c_1^{(\cdot)}, \dots, c_n^{(\cdot)}\right)$, where $(\cdot) \in \{\bot,\top\}$.
\autoref{fig:leipzig2} shows an example of using the MBB to abstract the running example in \autoref{fig:leipzig1}.

On the other hand,  \emph{space tiling} is an indexing technique for spatial data inspired by tessellation and previously used by LD optimization approaches such as \textsc{Orchid}~\cite{orchid} and $\mathcal{HR}^3$~\cite{NGON12}.
The main idea behind space tiling is to divide $n$-dimensional affine spaces into arbitrarily many hypercubes with the same edge length $\ell$.
These hypercubes are indexed with vectors $i \in \mathbb{N}^n$ to serve as addressable buckets for geometries.
In turn, the obtained index structures can be exploited by various optimization techniques.
We call $\Delta = \ell^{-1}$ the \emph{granularity factor}.
This notion of space tiling can be  generalized to hyperrectangles, in which case there exist $n$ independent granularity factors $\Delta_i$ where $ i \in \{0 \dots n\}$.
Note that although we eventually use hyperrectangles, we will stick to the term hypercube for the sake of simplicity and just define independent granularity factors when necessary.
\autoref{fig:leipzig3} shows our running example along with a grid of hypercubes using $\Delta=2$, where the green area will be indexed to each highlighted hypercube.




 \vspace{-0.7cm}
\section{Approach}
\label{sec:radon}
\vspace{-.4cm}

We have now introduced all ingredients necessary for defining the \radon algorithm (\autoref{alg:radon}).
\radon takes a set of source resources $S$, a set of target resources $T$ and a topological relation $r$ as input.
The goal of \radon  is to generate the mapping $M = \{(s,t) \in S \times T: r(s,t)\}$  \emph{efficiently}, where $r$ is a topological relation.
\radon addresses this challenge by means of three optimization steps:
\emph{Swapping} for index size minimization, \emph{space tiling} for indexing and \emph{filtering} to improve the runtime of the computation of topological relations.
In the following, we present each of these steps in detail.

\subsection{Swapping Strategy}
\label{sub:eth}

We introduce the \textit{Estimated Total Hypervolume} (\eh) of a set of geometries $X$ as
\begin{equation}
    \eh(X) = |X|\prod_{i=1}^d
    \frac{1}{|X|}\sum\limits_{x\in X} \left(\max\limits_{p \in x}\{\kappa_i(p)\} -\min\limits_{p \in x}\{\kappa_i(p)\}\right),
    \label{eq:eth}
\end{equation}
with $d$ being the number of dimensions of the resource geometries and $\kappa_i(p)$ denoting the coordinate of a point $p$ in the $i^\text{th}$ dimension.
If $\eh(T) < \eh(S)$, \radon swaps $S$ and $T$ and computes the reverse\footnote{Formally, the reverse relation $r'$ of a relation $r$ is defined as $r'(y, x) \Leftrightarrow r(x, y)$.}  relation $r'$ instead of $r$    (Lines~\ref{lin:swap1}--\ref{lin:swap2}). 
For example, if $r$ were the topological relation \texttt{covered} and $\eh(S) < \eh(T)$, then \radon swaps $T$ and $S$ and compute the reverse relation of $r$, i.e., \texttt{coveredBy}. 
The rationale behind using \eh instead of the size of the datasets is that even small datasets can contain very large geometries that span over a large number of hypercubes and would lead to large spatial index when used as source. 
For the sake of illustration, consider the running example in \autoref{fig:leipzig1}. Here, we can see that the \eh of NUTS (containing only the gray geometry) is greater than the \eh of CLC (containing the green and blue areas). Thus, we set $S=\text{CLC}$ and $T=\text{NUTS}$.

\subsection{Optimized Sparse Space Tiling}
\label{sub:optspacetiling}

In its second step, \radon utilizes space tiling to insert all geometries $s \in S$ and $t \in T$ into an index $I$, which maps resources to sets of hypercubes.
Let $\Delta_\varphi$ and $\Delta_\lambda$ be the granularities across the latitude and longitude (several strategies can be used to compute these values. We present and evaluate them in \autoref{sec:eval_granularity}). 
For indexing a resource $x$, we begin by computing its MMB's upper left and lower right corners coordinates $(\varphi_1(x), \lambda_1(x))$ and $(\varphi_2(x), \lambda_2(x))$ respectively (Line~\ref{lin:mbbs}).
Then, we map each $x$ to all hypercubes over which its MBB spans (Lines~\ref{lin:hypers1}--\ref{lin:hypers2}).
To this end, we transform the MBB's corner coordinates into hypercube indices using $\psi_\bot$ and $\psi_\top$ from \autoref{eq:psi}.
\begin{equation}
    \psi^\bot\left(x\right) = \floor{x\cdot\Delta_\varphi}\quad\quad\quad \psi^\top\left(x\right) = \ceil{x\cdot\Delta_\varphi}
    \label{eq:psi}
\end{equation}
We then map $x$ to all hypercubes with indices $(i,j)$ where $i,j \in \mathbb{Z}$,  $\psi^\bot(\varphi_1(x)) \leq i \leq \psi^\top(\varphi_2(x))$ and $\psi^\bot(\lambda_1(x)) \leq j \leq \psi^\top(\lambda_2(x))$.
Note that the special case of geometries passing over the \emph{antimeridian} is detected and dealt with by splitting such geometries into 2 geometries before and after the antimeridian. 
The index $I$ now contains the portions of the space (i.e., the hypercubes) within which portions of $x$ can potentially be found.
It is important to notice that entities in portions of space that do not belong to the hypercubes which contain elements of $S$ (denoted $I(S)$) will always be disjoint with the elements of $T$.
We leverage this insight as follows: 
We first index all $s \in S$. Then we follow the same procedure for $t \in T$ (Lines~\ref{lin:indext1}--\ref{lin:indext2}) but only index geometries $t$ that are potentially in hypercubes already contained in $I(S)$.
This \emph{optimized sparse space tiling} is the motivation for the previously introduced swapping strategy.
Indexing the dataset with the least \eh first results in an index $I$ with less hypercubes.

Consider again our running example in \autoref{fig:leipzig3} for the sake of illustration. Assume the granularity factors are $\Delta_{\varphi} = \Delta_{\lambda} =2$. The green area's MBB has the following corner coordinates:
    $(\varphi_1(g),\; \lambda_1(g))$ =   (12.340703846780286,\; 51.28797110806819) and 
    $(\varphi_2(g),\; \lambda_2(g))$ = 
  (12.389192648396918,\; 51.33902633403139).
Therefore, $\psi^\bot(\varphi_1(g))=24$, $\psi^\bot(\lambda_1(g))=102$, $\psi^\top(\varphi_2(g))=25$, $\psi^\top(\lambda_2(g))=103$ and thus this geometry will be indexed into the four highlighted hypercubes with index vectors $(24,102), (24,103), $ $(25,102)$ and $(25,103)$.
In \autoref{fig:leipzig4}, we highlighted all hypercubes containing the gray geometry after the optimized sparse space tiling.
Notice that many hypercubes are empty as a result of not containing any portion of the other dataset's geometries.

\subsection{Link Generation}
\label{sub:linkgeneration}
After the computation of the index $I$, \radon implements the last speedup strategy using a MBB-based filtering technique.
For each hypercube with indexed geometries from both $S$ and $T$ (Line~\ref{lin:hypercube}), \radon first discards unnecessary computations using the \textsc{TestMBB} procedure.
\textsc{TestMBB} optimizes the subset of DE-9IM relations for relations where one geometry has interior or boundary points in the exterior of the other geometry, i.e. $s\subseteq t$ or $t\subseteq s$ (e.g. \texttt{equals}, \texttt{covers}, \texttt{within} formally defined in \autoref{sec:setup}).
Let $\Box(g)$ denote the MBB geometry of a geometry $g$.
Note that $g\subseteq\Box(g)$  always holds.
We can now infer $\lnot r(\Box(s), \Box(t)) \Rightarrow \lnot r(s,t)$ using the transitivity of $\subseteq$.
For all other relations, \textsc{TestMBB} simply returns \textit{true}.
For example, in our running example in \autoref{fig:leipzig2}, if $r$ is the \texttt{within} topological relation, we do not need to compute $r$ for the blue geometry, as its MBB is not completely within the gray geometry's MBB.
In case the \textsc{TestMBB} method returns true, \radon carries out the more expensive computation of the topological relation between the geometries $s$ and $t$ (Line~\ref{lin:apicall}).
If $r(s,t)$ holds, \radon adds the pair $(s,t)$ to the result mapping $M$.
To make sure that we compute each pair $(s,t)\in S\times T$ at most once, we use a cache in form of a mapping $C$ which stores the already computed pairs of $(s,t)$ (Lines~\ref{lin:computedmap}-\ref{lin:computedmap2}).

\begin{proposition}
\radon is complete and correct.
\end{proposition}
\begin{proof}
\label{pro:complete}
Assume that we have two geometries $g_1$ and $g_2$. 
Assume that any of the entries of the DE-9IM matrix is true. 
Then, $g_1 \cap g_2 \neq \emptyset$.
Now given that $g_1 \subseteq \Box(g_1) \wedge g_2 \subseteq \Box(g_2)$, we can infer that $g_1 \cap g_2 \neq \emptyset \Rightarrow \Box(g_1) \cap \Box(g_2) \neq \emptyset$. Hence, checking MBBs guarantees that we find all pairs of geometries with $g_1 \cap g_2 \neq \emptyset$. This shows the \emph{completeness} of \radon.
The proof of the \emph{correctness} of \radon is trivial and is a direct result of the use of the call in Line~\ref{lin:apicall}, where \radon checks the pairs $(s, t)$ for whether $r(s,t)$ holds.
\end{proof}

\begin{algorithm}[htb]
\SetKwInOut{Input}{input}\SetKwInOut{Output}{output}
\Input{$S$, set of source resources.
	   $T$, set of target resources. 
	   $r$, topological relation.
      }
\Output{$M$, Mapping from $s \in S$ to $t \in T$ where $r(s,t)$ holds.}
    reversed $\leftarrow false$\;
    \If{$\eh(T) < \eh(S)$}{           \label{lin:swap1}
        \textsc{swap}($S,T$)\;
        r $\leftarrow r'$\;    
        reversed $\leftarrow true$\;    
    }                                                           \label{lin:swap2}
    \tcc{Get index $I$ using optimized sparse space tiling}
        $(\Delta_\varphi, \Delta_\lambda) \leftarrow \textsc{FindBestGranularity}(S,T)$\;     \label{lin:best_granularity}
    \ForEach{geometry $s \in S$}{
        $(\varphi_1(s), \lambda_1(s), \varphi_2(s), \lambda_2(s)) \leftarrow \textsc{GetMBBDiagonalCorners}(s)$\; \label{lin:mbbs}
        \For{$i \leftarrow \floor{\varphi_1(s) \cdot \Delta_\varphi} $ \KwTo $\ceil{\varphi_2(s) \cdot \Delta_\varphi}$}{
        \label{lin:hypers1}
            \For{$j \leftarrow \floor{\lambda_1(s) \cdot \Delta_\lambda} $ \KwTo $\ceil{\lambda_2(s) \cdot \Delta_\lambda}$}{
                \textsc{InsertIntoHypercube}($I(S), i, j, s$)\;  \label{lin:hypers2}
                $j \leftarrow j+1$\;
            }
            $i \leftarrow i+1$\;
        }
    }
    \ForEach{geometry $t \in T$}{                                   \label{lin:indext1}
        $(\varphi_1(t), \lambda_1(t), \varphi_2(t), \lambda_2(t)) \leftarrow \textsc{GetMBBDiagonalCorners}(t)$\;
        \For{$i \leftarrow \floor{\varphi_1(t) \cdot \Delta_\varphi} $ \KwTo $\ceil{\varphi_2(t) \cdot \Delta_\varphi} $}{
            \For{$j \leftarrow \floor{\lambda_1(t) \cdot \Delta_\lambda} $ \KwTo $\ceil{\lambda_2(t) \cdot \Delta_\lambda} $}{
                \If{$\textsc{GetHypercube}(I(S), i, j)$ is not empty }{
                    \textsc{InsertIntoHypercube}($I(T), i, j, t$)\;
                }
            $j \leftarrow j+1$\;
        }  
            $i \leftarrow i+1$\;
        }
    }                                                           \label{lin:indext2}
    \tcc{Generate Links}
  	\ForEach{hypercube $H_S \in I(S)$}{
    $H_T \leftarrow \textsc{GetHypercube}(I(T),  \varphi(H_S), \lambda(H_S))$\;
  	\If{$H_T$ is not empty}{                            \label{lin:hypercube}
  	    \For{$s \in H_S$}{
  	        \For{$t \in H_T$}{
  	            
                \If{$(s,t)\notin C$}{
                 \label{lin:computedmap}   
                $C\leftarrow C\cup \{(s,t)\}$\; 
                 \label{lin:computedmap2} \If{\textsc{TestMBB}$(r,(\varphi_1(s), \lambda_1(s), \varphi_2(s), \lambda_2(s)),(\varphi_1(t), \lambda_1(t), \varphi_2(t), \lambda_2(t)))$}{
      	                    \If{$r(s,t)$ is true}{ \label{lin:apicall}
      	              $M\leftarrow M\cup \{(s,t)\}$\; 
      	                    }
      	                }
  	                }
  	            }
  	        }
  	    }
    }
\If{reversed}{
    \Return{M'}\;                                   \label{lin:mpar}
}\Else{
    \Return{$M$}\;                                  \label{lin:m}
}
\caption{\radon -- Rapid Discovery of Topological Relations.}
\label{alg:radon}
\end{algorithm}

\afterpage{\clearpage}

\section{Evaluation}
\label{sec:evaluation}


In the following, we begin by introducing the relations and datasets as well as the hardware setting we used for carrying out our experiments in \autoref{sec:setup}.
In \autoref{sec:eval_granularity}, we evaluate different granularity selection policies for \radon.
Finally, we evaluate \radon vs. the LD framework \silk~\cite{Panayiotisldow2016} and the semantic spatiotemporal RDF store of \strabon~\cite{DBLP:conf/semweb/KyzirakosKK12}.

\subsection{Experimental Setup}
\label{sec:setup}

\subsubsection{Topological relations}
Only a subset of the topological relations obtainable through DE-9IM reflects the semantics of the English language~\cite{Clementini1993,clementini1994modelling} including \texttt{equals}, \texttt{with\-in}, \texttt{contains}, \texttt{disjoint}, \texttt{touches}, \texttt{meets}, \texttt{covers}, \texttt{cov\-ered\-By}, \texttt{in\-ter\-sects}, \texttt{inside}, \texttt{crosses} and \texttt{overlaps}.
Note that some of these relations are \emph{synonyms} (e.g., $\texttt{touches}(x,y) \Leftrightarrow \texttt{meets}(x,y)$) while others are \emph{combinations} of more atomic relations,(e.g., $\texttt{equals}(x,y) \Leftrightarrow  \texttt{within}(x,y) \wedge \texttt{contains}(x,y)$).
Moreover, some relations are the \emph{reverse} of some other relation.
Hence, in this evaluation, we focused on the rapid computation of the $7$ topological relations \texttt{within}, \texttt{touches}, \texttt{overlaps}, \texttt{intersects}, \texttt{equals}, \texttt{crosses} and \texttt{covers} as these are very commonly used~\cite{Panayiotisldow2016,Clementini1993,clementini1994modelling} and implemented in the systems we compare against. 
These relations are formally defined as follows:



\begin{definition}
A geometry $g_1$ is \textbf{topologically equal} to a geometry $g_2$ iff their interiors intersect and no parts of the interior or boundary of one geometry intersects the exterior of the other. Formally, 
$(I(g_1) \cap I(g_2)) \wedge
\lnot(I(g_1) \cap E(g_2) \neq \emptyset) \wedge
\lnot(B(g_1) \cap E(g_2)\neq \emptyset) \wedge
\lnot(E(g_1) \cap I(g_2)\neq \emptyset) \wedge
\lnot(E(g_1) \cap B(g_2)\neq \emptyset)$.
\end{definition}

\begin{definition}
Two geometries $g_1$ and $g_2$ are \textbf{topological intersects} iff they have at least one point in common. Formally, 
$(I(g_1) \cap I(g_2)) \wedge
\lnot(I(g_1) \cup I(g_2)\neq \emptyset) \vee
\lnot(I(g_1) \cap B(g_2)\neq \emptyset) \vee
\lnot(B(g_1) \cap I(g_2)\neq \emptyset) \vee 
\lnot(B(g_1) \cap B(g_2)\neq \emptyset)$.
\end{definition}

\begin{definition}
A geometry $g_1$ is \textbf{topologically touched} a geometry $g_2$ iff they have at least one boundary point in common, but no interior points. Formally, 
$(\lnot(I(g_1) \cap I(g_2)\neq \emptyset) \wedge (I(g_1) \cap B(g_2)\neq \emptyset)) \vee (\lnot(I(g_1) \cap I(g_2)\neq \emptyset)  \wedge (B(g_1) \cap I(g_2)\neq \emptyset)) \vee
(\lnot(I(g_1) \cap I(g_2)\neq \emptyset) \wedge (B(g_1) \cap B(g_2)\neq \emptyset))$.
\end{definition}

\begin{definition}
A geometry $g_1$ \textbf{topologically crosses} a geometry $g_2$ iff they have some but not all interior points in common, and the dimension of the intersection is less than the the maximum dimension of the two input geometries. Formally, 
$(dim(I(g_1) \cap I(g_2))) < \max(dim(I(g_1)), dim(I(g_2))) \wedge (g_1 \cap g_2 \neq g_1\neq \emptyset) \wedge (g_1 \cap g_2 \neq g_2\neq \emptyset)$.
\end{definition}

\begin{definition}
A geometry $g_1$ \textbf{topologically overlaps} a geometry $g_2$ iff they have some but not all points in common, they have the same dimension, and the intersection of the interiors of the two geometries has the same dimension as the geometries themselves.
Formally, 
$(I(g_1) \cap I(g_2)\neq \emptyset) \wedge
(I(g_1) \cap E(g_2)\neq \emptyset) \wedge
(E(g_1) \cap I(g_2)\neq \emptyset)$ for surfaces and
$dim(I(g_1) \cap I(g_2)) = 1 \wedge
(I(g_1) \cap E(g_2)\neq \emptyset) \wedge
(E(g_1) \cap I(g_2)\neq \emptyset)$ for lines.
\end{definition}

\begin{definition}
A geometry $g_1$ is \textbf{topologically within} a geometry $g_2$ iff $g_1$ lies in the interior of $g_2$. Formally, $(I(g_1) \cap I(g_2)) \wedge
\lnot(E(g_2) \cup I(g_1)\neq \emptyset) \wedge
\lnot(E(g_2) \cap B(g_1)\neq \emptyset)$.
\end{definition}

\begin{definition}
A geometry $g_1$ \textbf{topologically covers} a geometry $g_2$ iff every point of the interior and boundary of $g_2$ is also a point of either the interior or boundary of $g_1$. Formally, 
$((I(g_1) \cap I(g_2)\neq \emptyset) \wedge \lnot(E(g_1) \cap I(g_2)\neq \emptyset) \wedge \lnot(E(g_1) \cap B(g_2)\neq \emptyset)) \vee
((I(g_1) \cap B(g_2)\neq \emptyset) \wedge \lnot(E(g_1) \cap I(g_2)\neq \emptyset) \wedge \lnot(E(g_1) \cap B(g_2)\neq \emptyset)) \vee
((B(g_1) \cap I(g_2)\neq \emptyset) \wedge \lnot(E(g_1) \cap I(g_2)\neq \emptyset) \wedge \lnot(E(g_1) \cap B(g_2)\neq \emptyset)) \vee
((B(g_1) \cap B(g_2)\neq \emptyset) \wedge \lnot(E(g_1) \cap I(g_2)\neq \emptyset) \wedge \lnot(E(g_1) \cap B(g_2)\neq \emptyset))$.
\end{definition}

We dub the blue, green and gray areas in \autoref{fig:leipzig1} $a_1$, $a_2$ and $a_3$ respectively. Then, \texttt{distinct}$(a_1, a_2)$, \texttt{within}$(a_2, a_3)$ and \texttt{in\-ter\-sects}$(a_2, a_3)$,
hold.

\subsubsection{Datasets}
We evaluated our approach using two real-world datasets. 
The first dataset, the \emph{Nomenclature of Territorial Units for Statistics} or simply \emph{NUTS}\footnote{Version 0.91  (\url{http://nuts.geovocab.org/data/0.91/}) is used in this work.} is manually curated by 
the \emph{Eurostat} group of the European Commission.
NUTS contains a detailed hierarchical description of statistical regions for the whole European regions.
The second dataset, the \emph{CORINE Land Cover} or simply \emph{CLC} is an activity of the \emph{European Environment Agency} that collects data regarding the land cover of European countries. 
CLC contains 44 sub-datasets ranging from major categories of land cover (e.g., agricultural areas) to very specific characterisations (e.g., olive grives). 
Subsets of CLC range in size from $240$ to $248,242$ resources.\footnote{For more details about CLC see \url{https://datahub.io/dataset/corine-land-cover}}
For testing the scalability of \radon, we merged all subsets of CLC into one big dataset of size $2,209,538$ (dubbed $CLC_m$).
We preprocessed the datasets in the following fashion: To enable the processing of the NUTS dataset by \radon, \silk and \strabon, the \texttt{ngeo:posList} serialisation was converted into the WKT format prior to experiments. 
Moreover, because of a \silk issue\footnote{\url{https://github.com/silk-framework/silk/issues/57}}, we had to trim lines larger than $64$\,KB from all datasets in order to get a fair comparison. 
All the reported dataset sizes are after preprocessing.

\subsubsection{Hardware and Software}
All experiments were carried out on a 64-core $2.3$ GHz PC running OpenJDK 64-Bit Server 1.7.0 75 on \emph{Ubuntu} $14.04.2$ LTS. 
Unless stated otherwise, each experiment was assigned 20 GB RAM and a timeout limit of 2 hour. Experiments which ran longer than this upper limit were terminated and the
processed data percentage as well as the estimated time are reported.
For \silk experiments, we  ran  our  experiments  using its latest version (v2.6.1) with a blocking factor of $10$ as in~\cite{Panayiotisldow2016}.
For \strabon, we also used the latest version (v3.2.10)  with the accordingly tuned \emph{PostgreSQL} (v9.1.13) and \emph{PostGIS} (v2.0) as proposed by the developers. 
\radon is implemented as a part of the LD framework \limes.
A more complete list of results can be obtained from the project website\footnote{Link is omitted not to violate the blind review requirements.}.
Note that \radon achieves a precision, a recall and an F-measure of 1 by virtue of its completeness and correctness. \silk and \strabon theoretically achieve the same F-measure (we were not always able to check this value for the two systems as the experiments did not always terminate before the timeout).

\subsection{Experimental Results}
 

\subsubsection{Granularity Factor Selection Heuristic}
\label{sec:eval_granularity}

The aim of this experiment was to evaluate different heuristics to approximate the optimal \emph{granularity factors} $\Delta_\varphi$ and $\Delta_\lambda$ used for tiling the space and generating the sparse index of hypercubes. 
We tried 4 different heuristics corresponding to a statistical measure: \emph{minimum, maximum, median} and \emph{average}.
Each heuristic first computes the respective statistical measure $\eta$ independently for both datasets and both dimensions, resulting in 4 temporary values $h_{\eta,\varphi}(S)$, $h_{\eta,\varphi}(T)$, $h_{\eta,\lambda}(S)$, $h_{\eta,\lambda}(T)$. 
Finally, the granularity factor in \emph{each dimension} is the average of the two datasets.
Formally, 
\begin{align}
    &h_{\eta,\varphi}(X) = \heureta\limits_{x \in X} \left\{ \max\limits_{p \in x}\{\varphi(p)\} -\min\limits_{p \in x}\{\varphi(p)\} \right\} &
    &h_{\eta,\lambda}(X) = \heureta\limits_{x \in X} \left\{\max\limits_{p \in x}\{\lambda(p)\} -\min\limits_{p \in x}\{\lambda(p)\}\right\}\\
    &\Delta_{\eta,\varphi}(S,T) = \frac{1}{2}\Big(h_{\eta,\varphi}(S)+h_{\eta,\varphi}(T) \Big)&
    &\Delta_{\eta,\lambda}(S,T) = \frac{1}{2}\Big(h_{\eta,\lambda}(S)+h_{\eta,\lambda}(T) \Big)
\end{align}
Here $S, T$ are the input source and target datasets, $\varphi(p)$ the latitude of a point $p$, $\lambda(p)$ the longitude of $p$ and $\eta \in \left\{ \textit{min, max, avg, median} \right\}$.
We used all the $44$ subsets of the CLC dataset as input for this experiment and recorded how many times each heuristic achieved the best runtime for the \texttt{intersects} relation.
Additionally, when a heuristic was not the best in a run, we computed the percentage it was worse than the best one. 
The \emph{average} heuristic achieves the best result $24$ times out of $44$ experiments. Runner-up is \emph{median}, achieving the best runtime $17$ times.
Finally, the \emph{min} and \emph{max} heuristics achieved only $2$ and $1$ time(s) respectively.
Interestingly, \emph{average} and \emph{median} were only $4\%$ slower than the best measure on average when not being the best, while \emph{min} and \emph{max} where  $34\%$ and $61\%$ worse on average respectively.
Based on these results, we used the \emph{average} heuristic as the granularity selection policy in the rest of the experiments.



The basic idea behind the first three sets of of experiments is to quantify the \emph{speedup} gained by \radon over other LD frameworks. 
To the best of our knowledge, only the \silk LD framework recently~\cite{Panayiotisldow2016} implemented a multi-dimensional blocking approach to compute the topological relations.  Therefore, we compare \radon's and \silk's runtimes in the subsequent experiments  

In the \emph{first set of experiments}, we aimed of quantify the speedup of \radon over the other state-of-the-art approaches when applied to small datasets.
To this end, we ran $44$ experiments for each of the $7$ basic topological relations identified in the previous section. In each experiment, we compared one of the $44$ subsets of the CLC with the full NUTS. Altogether, we carried out $308$ experiments. Note that both \radon and \silk were ran on 1 core.
\radon achieves an average speedup of 
$221.52$, $213.76$, $4.94$, $4.82$, $4.77$, $4.76$ and $4.75$ for the relations 
\texttt{within}, \texttt{equals}, \texttt{covers}, \texttt{overlaps}, \texttt{intersects}, \texttt{crosses} and \texttt{touches} respectively.
Overall, \radon was able to outperform \silk by being $65.62$ times faster on average over all topological relations. 
Moreover, \radon was able to achieve a linear speedup relative to the dataset sizes. 
In \autoref{fig:speedup_silk}, we show an overview of a subset of the experimental results (including a linear fit) achieved on the relations on which \radon achieved the best (up to two orders of 450 times faster) and the poorest (up to 6.5 times faster) relative performance w.r.t. \silk.
Moreover, \radon ran significantly less complete computations of the relations at hand.
On average, $449$ times less computations per relation  (\autoref{fig:api}). 


\begin{figure}[tb]
    \centering
    \subfigure[\texttt{within}]{\label{fig:within}\includegraphics[trim={0 0 0 20},clip,width=0.5\textwidth]{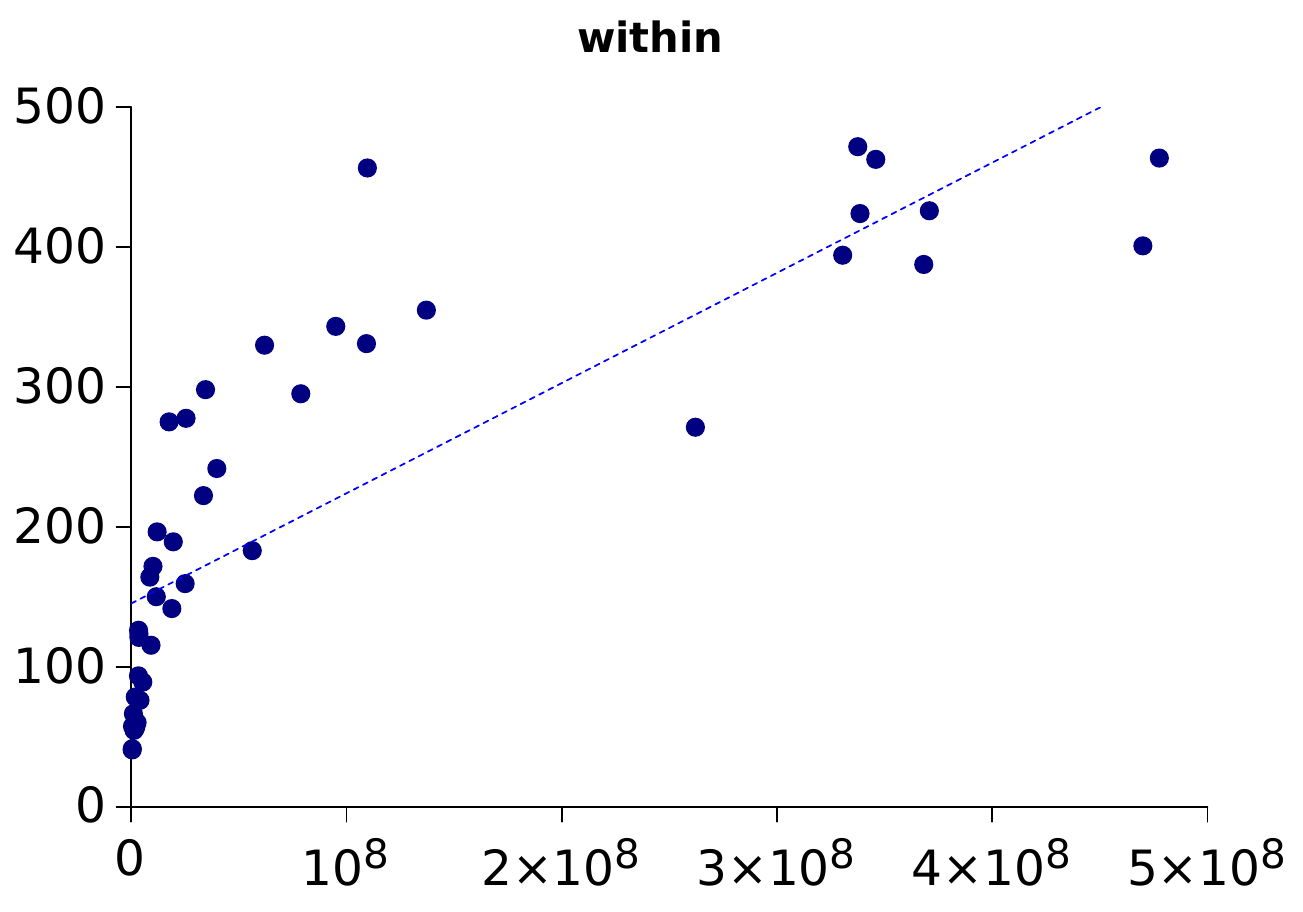}}~
    \subfigure[\texttt{touches}]{\label{fig:touches}\includegraphics[trim={0 0 0 20},clip,width=0.5\textwidth]{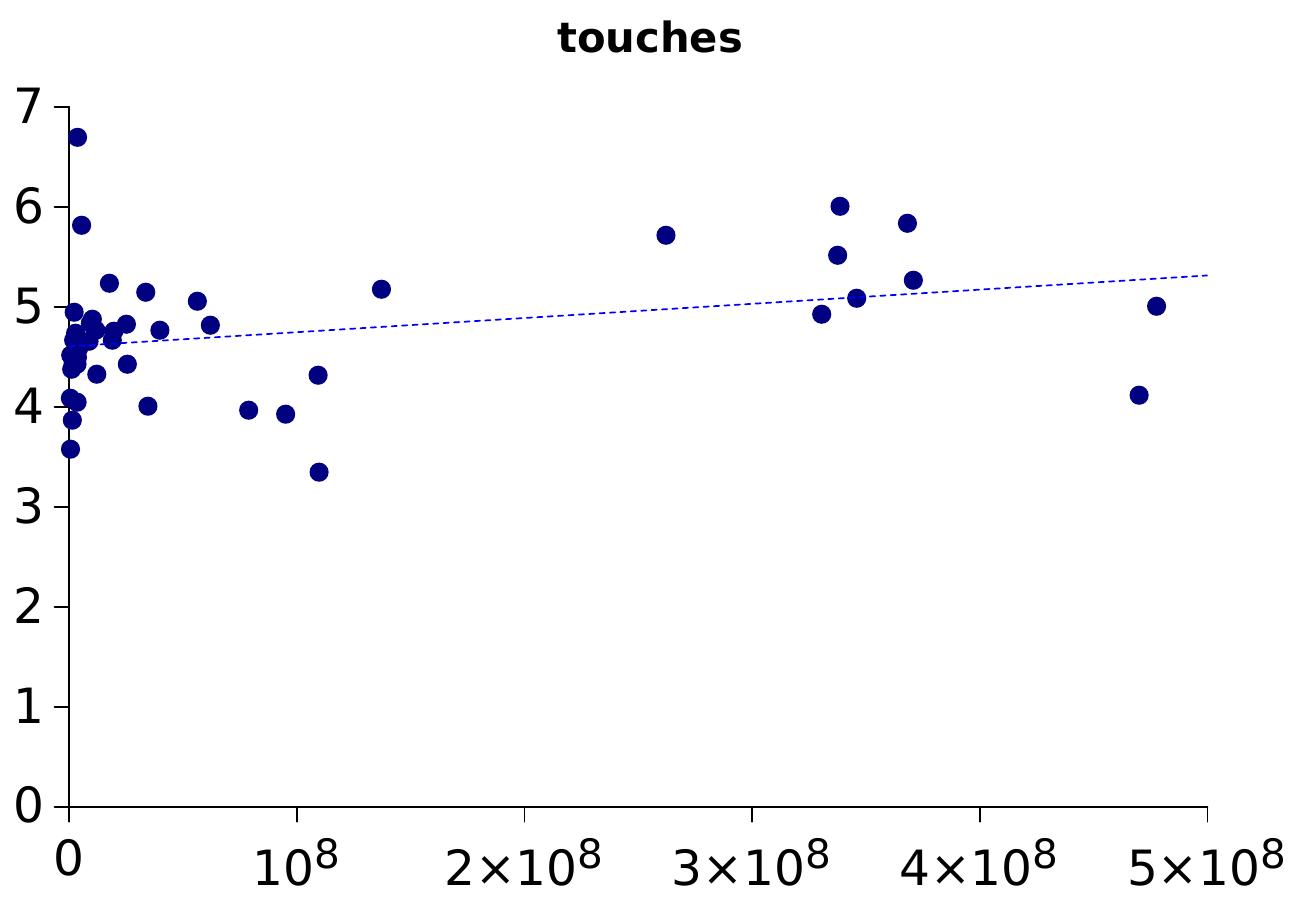}}
    \caption{Speedup of \radon over \silk. The x-axis represents the dataset sizes, y-axis represents the speedup. The blue dashed line is the linear regression line.}\label{fig:speedup_silk}
\end{figure}

In the \emph{second set of experiments} we aimed to evaluate the scalability of \radon when applied to big datasets.
Thus, we used the merged dataset $CLC_m$ as both source and target dataset and ran \radon and \silk on 1 core.
The results are shown in \autoref{tab:complete_clc}. \radon is able to finish all the tasks within $67.44$ minutes on average (maximum = $95.10$ minutes for the  \texttt{crosses} relation). 
On the other side, \silk was only able to (in average) finalize $0.34\%$ of each task within the 2 hour timeout limit.
We extrapolated the runtime of \silk linearly to get an approximation of how long it would need to carry out the tasks at hand. 
On average, \silk would need $24.85$ days to complete each task (linear extrapolation).
Consequently, \radon is at least $715.16$ times faster than \silk on average.
These results emphasize the ability of our algorithm to deal with large datasets even  when ran on 1 core.

\begin{table}[!h]
 \caption{Parallel implementation of \radon vs. \silk single machine for $CLC_m$ deduplication. 
      Runtimes are in minutes with timeout limit of 2 hour. Processes run above this upper limit were terminated and the processed data percentage as well as the estimated time are reported.}
   \centering
\begin{tabularx}{0.61\linewidth}{@{} lcrrr @{}}
\toprule
    Relation & \#Thr. & \radon & \multicolumn{1}{c}{\silk}  & Speedup\\ \midrule
    \multirow{4}{*}{equals} 
        & 1 & 24.11 & 36500 (0.33\%) & 1,513.58 \\
        & 2 & 13.15 & 21667 (0.55\%) & 1,647.58\\
        & 4 & 6.81 & 11750 (1.02\%) & 1,725.77\\
        & 8 & 3.79 & 6286 (1.91\%) & 1,658.78\\ \midrule
    \multirow{4}{*}{intersects}
        & 1 & 93.17 & 37500 (0.32\%) & 402.50 \\
        & 2 & 49.03 & 20667 (0.58\%) & 421.53\\
        & 4 & 25.11 & 12000 (1.00\%) & 477.81\\
        & 8 & 13.04 & 6300 (1.90\%) & 483.24\\ \midrule
    \multirow{4}{*}{crosses}
        & 1 & 95.10 & 35000 (0.34\%) & 368.05 \\
        & 2 & 48.02 & 21029 (0.57\%) & 437.96\\
        & 4 & 25.06 & 11881 (1.01\%) & 474.03\\
        & 8 & 13.08 & 6267 (1.91\%) & 479.21\\ \midrule
    \multirow{4}{*}{overlaps}
        & 1& 93.13 & 35000 (0.34\%) & 375.81 \\
        & 2 & 48.17 & 21404 (0.56\%) & 444.34\\
        & 4 & 25.09 & 11650 (1.03\%) & 464.32\\
        & 8 & 13.30 & 6235 (1.92\%) & 468.71\\ \midrule
    \multirow{4}{*}{within}
        & 1 & 36.47 & 35000 (0.34\%) & 959.74 \\
        & 2 & 18.26 & 20667 (0.58\%) & 1,131.86\\
        & 4 & 9.44 & 11765 (1.02\%) & 1,246.34\\
        & 8 & 5.92 & 6202 (1.93\%) & 1,048.34\\ \midrule
    \multirow{4}{*}{covers}
        & 1 & 35.62 & 36000 (0.33\%) & 1,010.75 \\
        & 2 & 18.51 & 21029 (0.57\%) & 1,136.10\\
        & 4 & 10.23 & 12000 (1.00\%) & 1,172.50\\
        & 8 & 5.33 & 6300 (1.90\%) & 1,182.13\\ \midrule
    \multirow{4}{*}{touches}
        & 1 & 94.50 & 35500 (0.34\%) & 375.68 \\
        & 2 & 47.71 & 22196 (0.54\%) & 465.18\\
        & 4 & 25.09 & 12121 (0.99\%) & 483.08 \\ 
        & 8	& 13.30	& 6381 (1.88\%) & 479.75 \\ \bottomrule
\end{tabularx}
\label{tab:complete_clc}
\end{table}

In the \emph{third set of experiments}, we wanted to quantify the \emph{speedup} gained by using a parallel implementation of \autoref{alg:radon} over the parallel implementation of \silk. 
For load balancing in \radon, we used the simple round robin load balancing policy~\cite{shreedhar1996efficient} with chunks size of $1000$.
As data, we used $CLC_m$ as both source and target.
The parallel implementations were configured to run using $2$, $4$ and $8$ threads.
The results (\autoref{tab:complete_clc}) show that our parallel implementation for \radon was able to discover all the topological relations in $20.83$ minutes in average (maximum of $49.03$ minutes in the case of the \texttt{intersect} relation).
On the other side, \silk implementation was only able to  (in average) finalize $1.16\%$ of each task within the 2 hours timeout limit.
We extrapolated the performance of \silk's parallel implementation and computed that it will need an average of $4.36$ days to finalize each task with $8$ threads.
Overall, our parallel implementation of \radon was up to $1725.77$ times  ($834.69$ times on average) faster than \silk.
Those results clearly show the scalability of \radon's parallel implementation.



\begin{lstlisting}[label={lst:intersects}, style=sparql, numbers=left, float=tb, numberstyle=\tiny, caption={SPARQL query for retrieving the \texttt{intersects} topological relation between resources from NUTS and CLC from \strabon.}]
SELECT ?s ?t WHERE {
    GRAPH <http://nuts.eu/>    { ?s geo:asWKT ?s_geometry. }
    GRAPH <http://clc.eu/#243> { ?t geo:asWKT ?t_geometry. }
    FILTER( strdf:intersects(?s_geometry, ?t_geometry) )
}
\end{lstlisting}


\begin{figure}[htb]
    \centering
    \subfigure[Average number of computations of topological relations]{\label{fig:api}\includegraphics[width=0.5\textwidth]{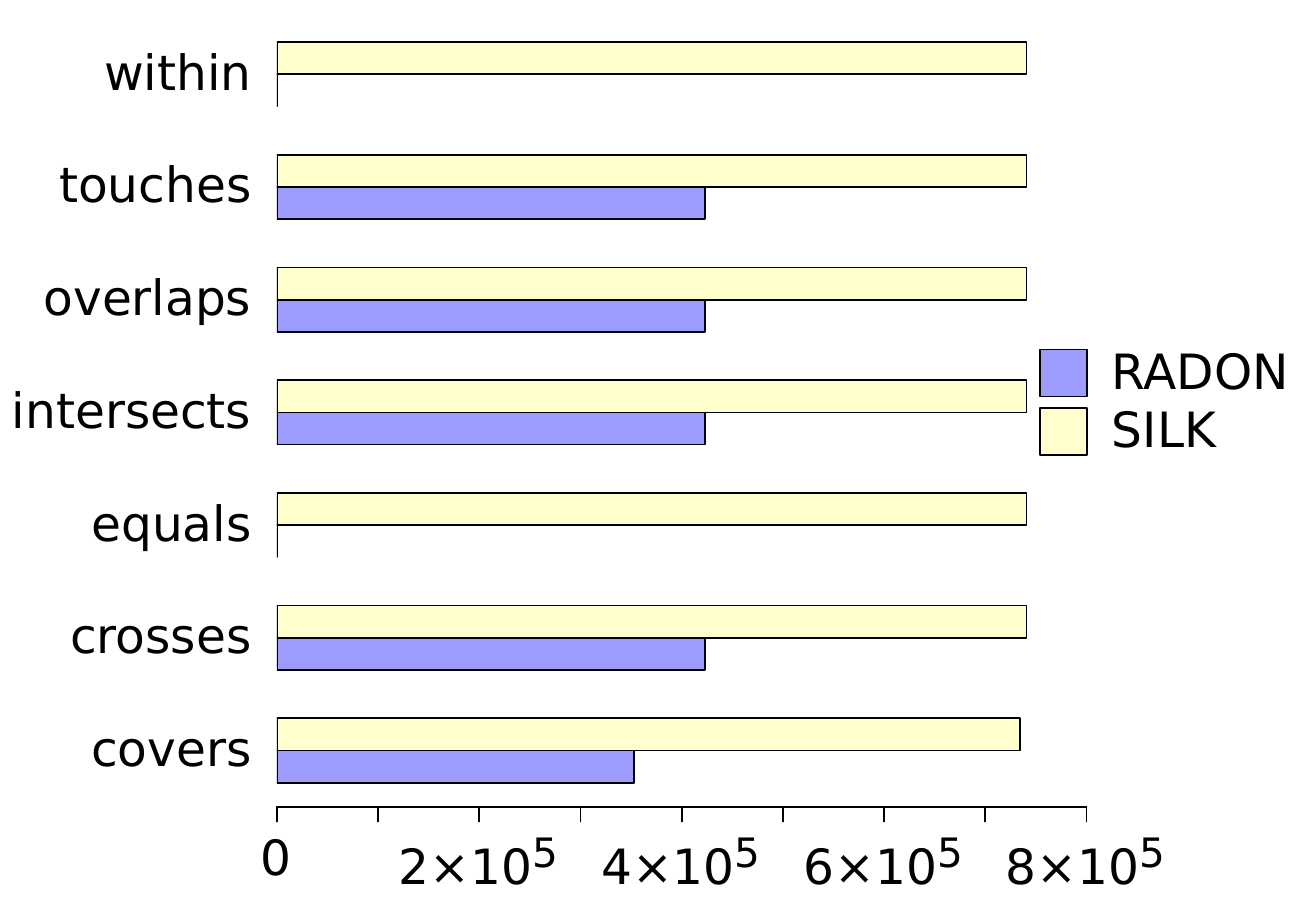}}~
    \subfigure[Average runtime]{\label{fig:runtime}\includegraphics[width=0.5\textwidth]{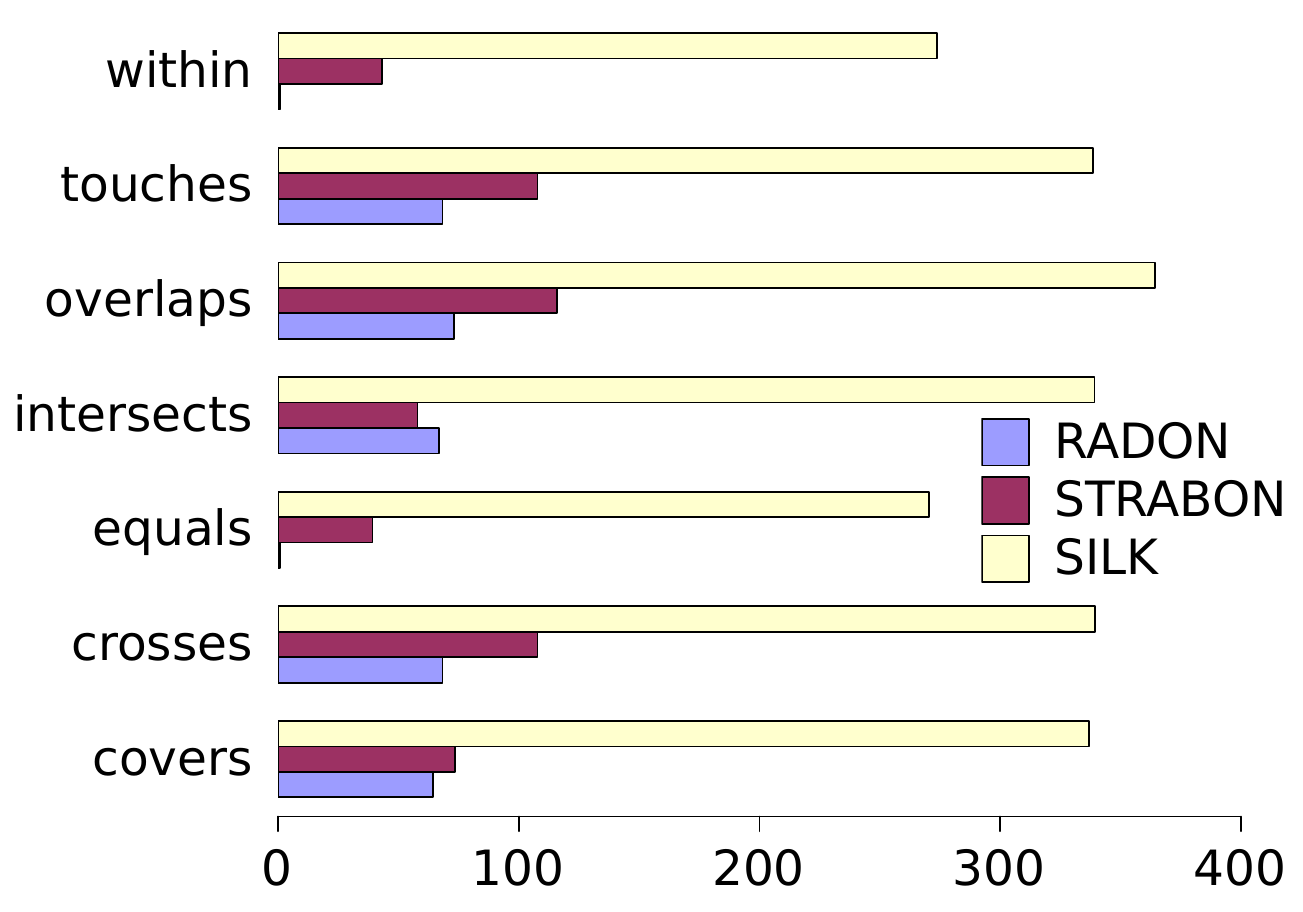}}
    \caption{Average number of complete computations of topological relations and average runtime for the datasets experiments. All runtimes are in seconds.}\label{fig:runtime_api}
\end{figure}

In our \emph{fourth set of experiments}, we aimed to compare \radon against \strabon on small datasets. 
The semantic spatio-temporal RDF store \strabon is not a LD framework but since it supports the \emph{GeoSPARQL} and \emph{stSPARQL} query languages. 
Therefore, \strabon can be employed for discovering topological relations via corresponding queries.
To compare with \strabon, we used the same setting we used in the first set of experiments.
\autoref{fig:runtime} shows the average runtimes result of both \radon and \strabon in seconds.
In average, \radon was $11.99$ times faster than \strabon. 
Interestingly, \strabon performed better than \radon on the \texttt{intersects} relation. 
The reason behind this behaviour is that \strabon uses an \emph{R-tree-over-GiST} spatial index over the stored geometries in the underlying \emph{PostGIS} database~\cite{DBLP:conf/semweb/KyzirakosKK12}.
This data structure is highly optimized for the retrieval of spatially connected objects.
Hence, \strabon requires solely a data retrieval to compute the \texttt{intersects} relation. However, this index is clearly outperformed by our sparse index in all the other relations as well as overall.

In our \emph{fifth and last set of experiments}, we evaluated the scalability of \radon vs. \strabon when tackling large datasets.
To this end, we applied the experimental setting we used in the second set of experiments ($S = T = CLC_m$).
\strabon was not able to finish any of the experiments within the 2-hour time limit while \radon required approx. 95.10 minutes in the worst case.
Given that \strabon provides no feedback pertaining to the progress of its tasks, we could not extrapolate its runtime.
Thus, we attempted a smaller deduplication experiment with only one subset of CLC, CLC-243, which is about $10$ times smaller than the merged  $CLC_m$ dataset.
Even these experiments did not finish within the $2$-hour limit.
Therefore, we approximated \strabon's runtime conservatively as follows:
Assume that the CLC-243 deduplication experiments would have finished just one minute after the $2$-hour timeout.
Assuming that \strabon's runtime scales linear with the input dataset size, the merged $CLC_m$ experiments would take roughly $20.17$ hours.
Having this overly optimistic estimate of \strabon's runtime, \radon achieves an average speedup of $24$. 
When we move from the assumption that \strabon scales linearly to the more realistic assessment that it scales in $O(n^2)$, then we get an average speedup of $241$.
Overall, our results show clearly that \radon outperforms the state of the art by up to 3 orders of magnitude in our experiments.



\section{Related Work}
\label{sec:relatedwork}

Based on the original works of Egenhofer et al.~\cite{egenhofer1991point}, Clementini et al.~\cite{clementini1994modelling} propose the The DE-9IM model to capture the topological relations in the $\mathbb{R}^2$.
In addition, the \emph{Simple Features Model} proposed by OGC\footnote{\url{http://www.opengeospatial.org/standards/sfs}} contain different subsets of the  topological relations that derive from the DE-9IM.
GeoSPARQL~\cite{geosparql} is a recent OGC standard that proposes a query language that enable the discovery of topological relations.
GeoSPARQL is implemented in the spatiotemporal RDF store \strabon~\cite{DBLP:conf/semweb/KyzirakosKK12}. Other frameworks such as Virtuoso\footnote{\url{http://virtuoso.openlinksw.com/}} and newly BlazeGraph\footnote{\url{https://www.blazegraph.com/}} support geo-spatial extensions of SPARQL. 
The discovery of topological relations has been paid little attention to in previous research related to Link Discovery~\cite{rw2013:introlinkeddata}.
Up to now, the state-of-the-art LD frameworks were able to discover only spatial similarities~\cite{geoLD2011,geoLD2006,geoLD2012}.
For example, \cite{orchid} uses the \emph{Hausdorff} distance to compute the point-set distance between geo-spatial entities.
In recent work, \cite{AGDISTIS_ECAI} implements an efficient approach for \emph{Allen} Relations extraction.
To the best of our knowledge, the only LD framework that support discovery of topological relaions is \silk \cite{Panayiotisldow2016}.
Based on \emph{MultiBlocking} technique, \cite{Panayiotisldow2016} computes the topological relations according to the DE-9IM standard between geo-spatial resources.
A detailed review of the current state of LD frameworks is recently published in~\cite{nentwig2015survey}.

\section{Conclusions and Future Work}
\label{sec:conclusion}
We presented \radon, an approach for rapid discovery of topological relations among geo-spatial resources.
\radon combines space tiling, minimum bounding box approximation and a sparse index to achieve a high scalability.
We evaluated \radon with real datasets of various sizes and showed that in addition to being complete and correct, it also outperforms the state of the art by up to three orders of magnitude (e.g., \texttt{equals} relation against \silk).
The parallel implementation of \radon currently employs a simple round robin load balancing policy. 
In future work, we aim to apply more sophisticated load balancing approaches, such as the particle-swarm-optimization based approaches~\cite{Sherif-dpso}.
In addition, we will consider the usage of other topology approximation methods, such as minimum bounding circles.
Finally, we will extend \radon to discover topological relations in higher dimensions, e.g., in 5D datasets.

\bibliographystyle{abbrv}
\bibliography{annex}
\end{document}